\documentclass{article}
\usepackage[utf8]{inputenc}
\usepackage[T2A]{fontenc}
\usepackage[english]{babel}
\inputencoding{utf8}

\usepackage{amsmath,amssymb,amsthm, nicefrac}
\usepackage{multicol}

\usepackage[
    pdfpagemode=UseNone,
    bookmarks=true,
    bookmarksopen=true,
    pdfpagelayout=SinglePage,
    pdfstartview={XYZ null null 1},
    pdfhighlight=/P,
    colorlinks=true,
    linkcolor=blue,
    citecolor=blue,
    urlcolor=blue
    ]{hyperref}
\usepackage{bookmark}
\bookmarksetup{
  open,
  numbered,
  depth=5, 
}
\usepackage{blkarray}
\usepackage{float}

\advance\textwidth 20mm  \advance\oddsidemargin -10mm
\advance\textheight 20mm \advance\topmargin -15mm

\def\C{\mathbb{C}}

\let\ds\displaystyle

\theoremstyle{plain}
\newtheorem{theorem}{Theorem}
\newtheorem{lemma}[theorem]{Lemma}
\newtheorem{proposition}{Proposition}

\newtheorem{assumption}{Assumption}
\newtheorem{remark}{Remark}

\DeclareMathOperator{\trace}{trace}

\newcommand{\brackets}[1]{\left( #1 \right)}

\newcommand{\LieBrackets}[1]{\left[ #1 \right]}

\DeclareMathOperator{\PI}{P_{1}}
\DeclareMathOperator{\PII}{P_{2}}
\DeclareMathOperator{\PIII}{P_{3}}
\DeclareMathOperator{\PIIIpr}{P_{3}^{\prime}}
\DeclareMathOperator{\PIV}{P_{4}}
\DeclareMathOperator{\PV}{P_{5}}
\DeclareMathOperator{\PVI}{P_{6}}

\newcommand{\Pkn}[1]{\text{P}_k^{ #1 }}

\newcommand{\PIIIprn}[1]{\text{P}_3^{\prime \, #1 }}

\newcommand{\PIIn}[1]{\text{P}_2^{ #1 }}
\newcommand{\PIn}[1]{\text{P}_1^{ #1 }}

\newcommand{\PPkn}[1]{\phantom{}_{ #1 }\text{P}_k}

\newcommand{\PPVn}[1]{\phantom{}_{ #1 }\text{P}_5}
\newcommand{\PPIVn}[1]{\phantom{}_{ #1 }\text{P}_4}

\newcommand{\PPIIIprn}[1]{\phantom{}_{ #1 }\text{P}_3^{\prime}}
\newcommand{\PPIIn}[1]{\phantom{}_{ #1 }\text{P}_2}

\newcommand{\PPPVIn}[1]{\phantom{}_{ #1 }\textbf{P}_6}
\newcommand{\PPPVn}[1]{\phantom{}_{ #1 }\textbf{P}_5}
\newcommand{\PPPIVn}[1]{\phantom{}_{ #1 }\textbf{P}_4}

\newcommand{\PPPIIIprn}[1]{\phantom{}_{ #1 }\textbf{P}_3^{\prime}}
\newcommand{\PPPIIn}[1]{\phantom{}_{ #1 }\textbf{P}_2}

\newcommand{\Painleve}{Painlev{\'e} }
\newcommand{\PPainleve}{Painlev{\'e}}

\makeatletter
\newcounter{savesection}
\newcounter{apdxsection}
\renewcommand\appendix{\par
  \setcounter{savesection}{\value{section}}%
  \setcounter{section}{\value{apdxsection}}%
  \setcounter{subsection}{0}%
  \gdef\thesection{\@Alph\c@section}}
\newcommand\unappendix{\par
  \setcounter{apdxsection}{\value{section}}%
  \setcounter{section}{\value{savesection}}%
  \setcounter{subsection}{0}%
  \gdef\thesection{\@arabic\c@section}}
\makeatother

\usepackage{mathtools}
\mathtoolsset{showonlyrefs,showmanualtags}

\allowdisplaybreaks

\title{On classification of non-abelian Painlev\'e~type~systems}

\date{}

\author{I.A. Bobrova\thanks{National Research University Higher School of Economics, Moscow, Russian Federation.},~ V.V. Sokolov\thanks{L.D.~Landau Institute for Theoretical Physics, Chernogolovka, Russian Federation.} }

\usepackage{tikz}

\begin{document}
\maketitle

\begin{abstract}
We find all non-abelian generalizations of $\PI-\PVI$ \Painleve systems such that the corresponding autonomous system obtained by freezing the independent variable is integrable. All these systems have isomonodromic Lax representations. 
\medskip

\end{abstract}

\section{Introduction}

The famous six \Painleve equations play a significant role in modern mathematical physics. The interest in their non-commutative extensions was motivated by the needs of modern quantum physics as well as by natural attempts of mathematicians to extend ``classical'' structures to the non-commutative case.

In the paper \cite{okamoto1980polynomial}, the \PPainleve \, equations were written as polynomial systems of the form
\begin{gather} \label{eq:fullansatz}
    \left\{
    \begin{array}{lcl}
         f(z) \, \displaystyle \frac{d u}{d z} 
         &=& P_1 (u, v) + z \, Q_1 (u, v),
         \\[3mm]
         f(z) \, \displaystyle \frac{d v}{d z}
         &=& P_2 (u, v) + z \, Q_2 (u, v)
         ,
    \end{array}
    \right.
\end{gather}
where $f(z)=z(z-1)$ for the $\PVI$ system, $f(z)=z$ in the $\PV$ and $\PIII$ cases and $f(z)=1$ for remaining types of \PPainleve \, systems. The right-hand sides of systems depend on several arbitrary parameters $\kappa_i$.

We consider systems of the form \eqref{eq:fullansatz},
where $P_i$, $Q_i$ are polynomials with {\it constant scalar} coefficients in non-commutative variables $u$ and $v$. For simplicity, we regard these variables as matrices of an arbitrary size $m$\footnote{Using the non-abelian formalism (see, for example, \cite{kontsevich1993formal}), all considerations can be trivially generalized to the non-abelian case. For this reason, in this paper we do not distinguish  ``matrix'' and ``non-abelian'' systems.}. We assume that if $m=1$ then the system \eqref{eq:fullansatz} coincides with one of the six \Painleve systems $\PI-\PVI$. Some examples of integrable non-abelian systems of the form \eqref{eq:fullansatz} are contained in \cite{ Kawakami_2015,Balandin_Sokolov_1998,Adler_Sokolov_2020_1,Bobrova_Sokolov_2022,Retakh_Rubtsov_2010,adler2020,Adler_Kolesnikov_2022}. Some of them were found using the existence of special isomonodromy representations or the  \PPainleve-Kovalevskaya test  while several of the systems have been derived from integrable PDEs and lattices by reductions.

Classification papers \cite{Balandin_Sokolov_1998,Adler_Sokolov_2020_1,
Bobrova_Sokolov_2022} devoted to systems of $\PI$, $\PII,$ and $\PIV$ were based on the matrix \PPainleve-Kovalevskaya test. Unfortunately, for the $\PIII$, $\PV,$ and $\PVI$ systems this approach turns out to be too laborious. Мoreover, in this approach it is essential that the variables $u$ and $v$ are matrices. In \cite{bobrova2022non} a different integrability criterion based on the existence of a non-abelian Okamoto integral was used. As a result, a collection of new non-abelian systems of $\PIII$, $\PV,$ and $\PVI$ type was obtained. However, some of known non-abelian systems do not have the Okamoto integral.

In this paper we propose an approach related to properties of 
the autonomous system
\begin{gather} \label{eq:noz}
    \left\{
    \begin{array}{lcl}
         \displaystyle \frac{d u}{d t} 
         &=& P_1 (u, v) + z \, Q_1 (u, v),
         \\[3mm]
         \displaystyle \frac{d v}{d t}
         &=& P_2 (u, v) + z \, Q_2 (u, v)
         ,
    \end{array}
    \right.
\end{gather}
obtained from system \eqref{eq:fullansatz} by freezing the independent variable. In system \eqref{eq:noz} 
we treat $z$ as a parameter. We call system \eqref{eq:noz} {\it the auxiliary autonomous system} for \eqref{eq:fullansatz}.

Our classification of integrable non-abelian \Painleve systems \eqref{eq:fullansatz} is based on the following 
\begin{assumption} 
\label{assumpt1}
The corresponding auxiliary system is supposed to be integrable. This means that{\rm:}
\begin{enumerate}
    \item   
    System \eqref{eq:noz} has a Lax representation 
    \begin{equation}\label{LMLax}
        L_t
        = [M, L]
    \end{equation} 
    with a spectral parameter{\rm.}
    \item   
    It has infinitely many polynomial non-abelian infinitesimal symmetries commuting with each other. Each of these symmetries
    has a Lax representation of the form \eqref{LMLax} with the same operator $L$ {\rm(}but with different operator $M${\rm)}.    
    \item  
    All systems of this hierarchy have the same set of first integrals defined by ${\rm trace}\,(L^i)$.
    \item 
    The Lax operators, first integrals, and symmetries coincide with the corresponding components of the scalar hierarchy when the size $m$ of matrices is equal to one.
  \end{enumerate}
\end{assumption}

It turns out that known examples of non-abelian \Painleve systems satisfy some of these conditions (for $\PI$ and $\PII$ see \cite{Rub}). The hierarchy of Assumption \ref{assumpt1} can be obtained from the trivial  hierarchy of the scalar auxiliary system by the non-abelianization procedure \cite{SW2,Bobrova_Sokolov_2022} described in Section \ref{nonabelianization}.

Using the simplest example of \PPainleve-2 systems, we show in Subsection \ref{sec:P2_case} how the ingredients of a non-abelian hierarchy can be constructed. It turns out
that non-abelianization of only two simplest first integrals and one infinitesimal symmetry allows one to obtain a finite list of non-abelian $\PII$ systems.
These systems coincide with ones from the paper \cite{Adler_Sokolov_2020_1}, where a different criterion of integrability for matrix \PPainleve-2 equations has been used.

In Section \ref{P6}, using the same approach, we find a complete list of \PPainleve-6 type autonomous systems. Due to the complexity of the calculations, we first classify homogeneous systems, which are the leading parts of autonomous systems of the $\PVI$ type. For each leading part, finding the remaining coefficients is not too difficult. The resulting list contains 35 different systems, including 19 systems from  \cite{bobrova2022classification,bobrova2022non}. 

One of the reasons for the appearance of a large number of non-abelian systems of the same type is the presence of transformations that preserve integrability of systems \eqref{eq:noz} or \eqref{eq:fullansatz}. The simplest example of such a transformation is the
matrix transposition 
\begin{align}\label{tau}
    &&
    \tau(u)
    &=u^T,
    &
    \tau(v)
    =v^T.
    &&
\end{align}

In Appendix \ref{sec:sysintlistP6}, we present an explicit form of only 8 new systems non-equivalent with respect to $\tau$. For the old 19 systems we use the notation and lists from our previous  papers \cite{bobrova2022classification,bobrova2022non}.

To reconstruct the  \Painleve system of the form \eqref{eq:fullansatz} corresponding to an autonomous $\PVI$ system  \eqref{eq:noz} described in Section \ref{P6}, we have to replace the time $t$ with $z$ and choose $f(z)=z (z-1)$ in \eqref{eq:fullansatz}. 

A group of  point transformations (see Subsection \ref{sec:trgroup_P6}) isomorphic to $S_3 \times \mathbb{Z}_2$ acts on the set of all non-autonomous systems of the \PPainleve-6 type. The Hamiltonian systems \cite{bobrova2022classification} form one orbit of this action, while the systems with the generalized Okamoto integral \cite{bobrova2022non} and systems listed in Appendix \ref{sec:sysintlistP6} have three and four orbits, respectively. Thus, there are 8 non-equivalent systems of the \PPainleve-6 type.

Following Assumption \ref{assumpt1}, in Section \ref{P5-P1}, we find  all autonomous non-abelian \Painleve systems of $\PI - \PV$ types. All known and several new systems are contained in our list. New systems are presented in Appendices \ref{sec:sysintlistP2} -- \ref{sec:sysintlistP5}. Non-autonomous \Painleve $\PI - \PV$ systems are restored using autonomous systems found in Section \ref{P5-P1}.

Section \ref{pensyst} is devoted to the study of non-autonomous \Painleve systems whose auxiliary autonomous systems were constructed in Sections \ref{P6} and \ref{P5-P1}.
 In subsection \ref{Sec5.1}, we describe admissible transformations in each of the cases $\PI - \PVI$. 
For non-abelian systems of $\PVI$ type  we present in subsection \ref{Sec5.2.1} an 
universal formula for isomonodromic Lax pairs of the form
\begin{gather} \label{eq:zerocurvcond}
    \mathbf{A}_z - \mathbf{B}_{\lambda}
    = [\mathbf{B}, \mathbf{A}],
\end{gather} 
where $\mathbf{A} (z, u,v, \lambda)$ and $\mathbf{B} (z, u,v, \lambda)$ are some $2\times 2$ matrices, rational in the spectral parameter $\lambda$ and in $z$, linear in the parameters $\kappa_i$, and polynomial in $u,v$. It is suitable for a wide family of systems whose autonomous systems satisfy the condition 3 of Assumption \ref{assumpt1}.  Besides $\kappa_i,$ the family contains $9$ parameters. To specify these parameters we apply the non-abelianization procedure for symmetries of the autonomous system. Similar universal formulas for the Lax representation were also found in Subsections \ref{Sec5.2.2} -- \ref{Sec5.2.6} for $\PV-\PI$ type systems.

It turns out that non-abelian systems of the \PPainleve-6 type generate non-abelian \Painleve systems of the $\PI-\PV$ type by the standard limiting procedure described in \cite{bobrova2022non}.
  In Subsection \ref{sec:deg} we present limiting transition schemes for new $\PVI$ systems obtained in this paper.

\section{Non-abelianization of hierarchy}
\label{nonabelianization}
In the scalar \Painleve case, any auxiliary system \eqref{eq:noz} has a polynomial Hamiltonian $H$ and a Lax representation \eqref{LMLax}. The hierarchy of commuting symmetries consists of the systems
\begin{equation}\label{scsymgen}
    \left\{
    \begin{array}{lcl}
    u_{t_{k}} &=& (P_1+z \, Q_1) \, H^k,
    \\[2mm]
    v_{t_{k}} &=& (P_2+z \, Q_2) \, H^k,
    \end{array}
    \right.
\end{equation}
which have the Lax representations 
\begin{equation}\label{LMLaxsym}
        L_{t_{k}}
        = [H^k M, L].
\end{equation} 
These symmetries are trivial: if we fix a value of the integral $H$ then the symmetries coincide (up to a scaling of $t$) with the system \eqref{eq:noz}. 

We assume that all the above components of the hierarchy
admit the non-abelianization. This means that 
\begin{itemize}
    \item we consider systems of the form \eqref{eq:noz} with matrix polynomials $P_i$ and $Q_i$ such that if the size $m$ of matrices $u$ and $v$ is equal to one, the systems coincide with the corresponding auxiliary autonomous system for a given scalar system of \Painleve type;
    
    \item there exist matrix polynomials $S_i$ such that if $m=1$ they coincide with $H^i$ and traces of $S_i$ are first integrals of the matrix system\footnote{We do not assume that $\trace S_1$ is its Hamiltonian.}; 
    
    \item the matrix system has a Lax representation \eqref{LMLax} that coincides with the representation for scalar system if $m=1$;
    
    \item for any $k$ the matrix system has a (matrix) symmetry that coincides with \eqref{scsymgen} if $m=1$. These symmetries commute with each other;
    
    \item any matrix symmetry possesses a Lax representation coinciding with the scalar representation \eqref{LMLaxsym} if $m=1$.
\end{itemize}

\subsection{\PPainleve-2 systems}
\label{sec:P2_case}
Let us illustrate our approach by the simplest example of the \PPainleve-2 systems. Starting with trivial hierarchy of the scalar system, we show that the non-abelianization of several its ingredients leads to a finite list of matrix \PPainleve-2 systems. It was proved in \cite{Adler_Sokolov_2020_1} that all these systems satisfy the \PPainleve-Kovalevskaya test.

In the commutative case, the system has the form
\begin{equation}
    \label{eq:sysP2}
    \left\{
    \begin{array}{lcl}
         \dfrac{d u}{d z}
         &=& - u^2 + v - \tfrac12 z,  
         \\[3mm]
         \dfrac{d v}{d z}
         &=& 2 u v + \kappa,
    \end{array}
    \right.
\end{equation}
where $\kappa$ is an arbitrary parameter.
This system is Hamiltonian with the Hamiltonian 
\begin{equation}
    \label{eq:hamP2}
    H
    = - u^2 v + \tfrac12 v^2 - \kappa u - \tfrac12 z v,
\end{equation} 
which is a first integral for the auxiliary autonomous system
\begin{equation}
    \label{eq:ausysP2}
    \left\{
    \begin{array}{lcl}
         \dfrac{d u}{d t}
         &=& - u^2 + v - \tfrac12 z,  
         \\[3mm]
         \dfrac{d v}{d t}
         &=& 2 u v + \kappa.
    \end{array}
    \right.
    \end{equation}
It is clear that the systems
\begin{equation}
    \label{eq:ausymP2}
    \left\{
    \begin{array}{lcl}\ds
         \frac{d u}{d t_{k}}
         &=& (- u^2 + v - \tfrac12 z)\, H^k,  
         \\[3mm]
         \ds \frac{d v}{d t_{k}}
         &=& (2 u v + \kappa) \, H^k
    \end{array}
    \right.
\end{equation}
are commuting infinitesimal symmetries for \eqref{eq:ausysP2}.
    
It is easy to verify that Lax representation \eqref{LMLax} with    
\begin{align}
    \label{eq:LM}
    \begin{aligned}
    L (\lambda)
    &= 
    \begin{pmatrix}
    -4 & 0 \\ 0 & 4
    \end{pmatrix}
    \lambda^2
    + 
    \begin{pmatrix}
    0 & 4 u
    \\ 
    4 u & 0
    \end{pmatrix}
    \lambda
    + 
    \begin{pmatrix}
    2 u^2 + z & - 2 u^2 + 2 v - z
    \\ 
    2 u^2 - 2 v + z & - 2 u^2 - z
    \end{pmatrix}
    + 
    \begin{pmatrix}
    0 & \kappa
    \\ 
    \kappa & 0
    \end{pmatrix}
    \lambda^{-1}
    ,
    \\[2mm]
    M (\lambda)
    &= 
    \begin{pmatrix}
    1 & 0 \\ 0 & -1
    \end{pmatrix}
    \lambda 
    + 
    \begin{pmatrix}
    - \beta u & - u
    \\ 
    - u & - \beta u
    \end{pmatrix}
    \end{aligned}
\end{align}
is a Lax representation for system  \eqref{eq:ausysP2}\footnote{In the scalar case the parameter $\beta$ in the $M$-operator is inessential.} and therefore the operators $L$ and $M_k=H^k\, M$ define a Lax pair for system \eqref{eq:ausymP2}.

The hierarchy described above looks completely trivial: all integrals are powers of only one integral $H$, the symmetries and corresponding $M$-operators are proportional if we fix a value of $H$. However, all these objects become non-trivial after the non-abelianization \cite{SW2,Bobrova_Sokolov_2022}.
 
For the non-abelianization of system \eqref{eq:ausysP2}, we replace each monomial in the right-hand side by a sum of non-commutative monomials such that under the commutative reduction the system coincides with \eqref{eq:ausysP2}.
The result can be written in the form
\begin{gather}
    \label{eq:sysP2_nc}
    \left\{
    \begin{array}{lcl}
          \dfrac{d u}{d t}
         &=& - u^2 + v - \tfrac12 z,  
         \\[2mm]
          \dfrac{d v}{d t}
         &=& u v + v u + \beta [v, u] + \kappa,
    \end{array}
    \right.
\end{gather}
where {$\kappa$}, $\beta \in \mathbb{C}$. 

\subsubsection{Non-abelianization of integrals}
\label{sec:P2int}

Let us generalize  to the non-abelian case several first integrals.
\begin{proposition}
\label{thm:trint2}
    For any $\beta$ the autonomous system \eqref{eq:sysP2_nc} possesses first integrals of the form $\trace S_1$ and $\trace S_2$ such that under the commutative reduction the non-commutative polynomials $S_1$ and $S_2$ coincide with $H$ and $H^2$, respectively.
\end{proposition}
\begin{proof}
     The general ansatz\footnote{Note that such polynomials are defined up to commutators.} for the polynomial $S_1$ can be written as
    \begin{align}
        S_1
        &= - u^2 v + \tfrac12 v^2 - \kappa u - \tfrac12 z v.
    \end{align}
    One can verify that the condition
    \begin{align}
    \label{eq:intcond_nc}
        \frac{d}{dt}\trace \brackets{ S_i}
        = 0
    \end{align}
    with $i = 1$ holds. 
    
    Using the explicit formula 
    \begin{align}
        H^2
        = u^4 v^2
        - u^2 v^3  
        + \tfrac14 v^4
        + 2 \kappa u^3 v
        - \kappa u v^2
        + \kappa^2 u^2
        + \, z \brackets{
        u^2 v^2 - \tfrac12 v^3 + \kappa u v
        }
        + \tfrac14 z^2 v^2,
    \end{align}
    we construct a general ansatz for the non-abelian polynomial $S_2$. Since  the monomials
    \begin{align}
        &&
        \{
        v u^3 v u, \,\,
        v u^2 v u^2, \,\,
        v u^4 v
        \},
        &&
        \{
        v u v u v, \,\,
        v u^2 v^2
        \},
        &&
        \{
        v u v u, \,\,
        v u^2 v
        \},
         &&
        \{
        v u v 
        \}
        &&
    \end{align}
    define complete lists of homogeneous polynomials of degrees $(4, 2),$ $(2,3),$ $(2,2)$ and $(1,2)$ in $u$ and $v$ up to commutators, we take
    \begin{align}
        S_2
        = a_1 v u^3 v u 
        + a_2 v u^2 v u^2 
        + (1 - a_1 - a_2) v u^4 v
        + a_3 v u v u v 
        + (- 1 - a_3) v u^2 v^2 
        + \tfrac14 v^4
        \\ 
        + \, 2 \kappa v u^3
        - \kappa v u v 
        + \kappa u^2
        + z \brackets{
        a_4 v u v u
        + (1 - a_4) v u^2 v
        - \tfrac12 v^3 
        + \kappa v u
        }
        + \tfrac14 z^2 v^2.
    \end{align}
     Condition \eqref{eq:intcond_nc} with $i = 2$ leads to a system of linear algebraic equations for $a_j,\, j=1,\dots,4$, whose coefficients depend on $\beta$. Solving this system, we find that for any $\beta$
    \begin{align}
        &&
        &&
        a_1
        &= 0,
        &
        a_2
        &= 1,
        &
        a_3
        &= 0
        &&
        &&
    \end{align}
    and $a_4$ remains to be arbitrary. 
    Therefore, the autonomous system has the first integral $\trace S_2$, where
    \begin{align}
        S_2
        = v u^2 v u^2 
        - v u^2 v^2 
        + \tfrac14 v^4
        + \, 2 \kappa v u^3
        - \kappa v u v 
        + \kappa u^2
        + z \brackets{
        v u^2 v
        - \tfrac12 v^3 
        + \kappa v u
        }
        \\
        + \, a_4 z (v u^2 v - v u v u) 
        + \tfrac14 z^2 v^2.
    \end{align}
\end{proof}
In the case of the non-commutative \PPainleve-2 system, the general ansatz for the right-hand side depends only on one parameter $\beta$ and our first test, based on the non-abelianization of the first integrals, does not determine $\beta$. For comparison, in the case of the \PPainleve-6, it reduces the number of unknown coefficients from 18 to 9.

\subsubsection{Non-abelianization of symmetries}
\label{Sec202}

In contrast to the above computation, the non-abelianization of the simplest symmetry \eqref{eq:ausymP2} with $k = 1$ gives rise to a condition for $\beta$. The general ansatz for the non-abelian symmetry  has the form
\begin{align}
    &
    \begin{aligned}
        u_{t_1}
        = a_1 u^4 v 
        + a_2 u^3 v u 
        + a_3 u^2 v u^2 
        + a_4 u v u^3 
        + \brackets{
        1 - \sum a_i
        } v u^4
        + c_1 u^2 v^2
        + c_2 u v u v
        + c_3 u v^2 u
        \\
        + \, c_4 v u^2 v
        + c_5 v u v u 
        + \brackets{
        - \tfrac32 - \sum c_i
        } v^2 u^2
        + \kappa u^3 
        + \tfrac12 v^3 
        + e_1 u v
        + (- \kappa - e_1) v u
        \\
        + \, z \brackets{
        g_1 u^2 v 
        + g_2 u v u
        + (1 - g_1 - g_2) v u^2
        - \tfrac34 v^2 + \tfrac12 \kappa u
        }
        + \tfrac14 z^2 v,
    \end{aligned}
    \\
    &
    \begin{aligned}
        v_{t_1}
        = b_1 \, u^3 v^2 
        + b_2 \, u^2 v u v
        + b_3 \, u^2 v^2 u
        + b_4 \, u v u^2 v
        + b_5 \, u v u v u
        + b_6 \, u v^2 u^2
        + b_7 \, v u^3 v
        + b_8 \, v u^2 v u
        \\
        + \, b_9 v u v u^2
        + \brackets{
        - 2 - \sum b_i
        } v^2 u^3
        + d_1 u v^3 
        + d_2 v u v^2
        + d_3 v^2 u v
        + \brackets{
        1 - \sum d_i
        } v^3 u
        \\
        + \, f_1 u^2 v
        + f_2 u v u
        + (- 3 \kappa - f_1 - f_2) v u^2 
        + \tfrac12 \kappa v^2 
        - \kappa^2 u
        \\
        + \, z \brackets{
        h_1 u v^2
        + h_2 v u v
        + (- 1 - h_1 - h_2) v^2 u 
        - \tfrac12 \kappa v
        },
    \end{aligned}
\end{align}
where $a_i$, $b_i$, $c_i$, $d_i$, $e_1$, $f_i$, $g_i$, $h_i \in \mathbb{C}$.
According to the item 3 of Assumption \ref{assumpt1}, the symmetry should have the first integrals $\trace S_1$ and $\trace S_2$ found in Proposition \ref{thm:trint2}. That leads to the following linear relations between coefficients of the symmetry
\begin{align}
    &\begin{aligned}
        b_1
        &= - a_1,
        &&&
        b_2 
        &= - a_1 - a_2,
        &&&
        b_3
        &= 0,
        &&&
        b_4 
        &= - a_1 - a_2 - a_3,
        &&&
        b_5
        &= 0,
    \end{aligned}
        \\[1mm]
    &\begin{aligned}
        b_6 
        &= 0,
        &&&
        b_7
        &= 1 - a_2 - a_3 - a_4,
        &&&
        b_8 
        &= - 1 + a_1 + a_2,
        &&&
        b_9
        &= - 1 + a_1 + a_2 + a_3,
    \end{aligned}
        \\[1mm]
    &\begin{aligned}
        c_4 
        &= - \tfrac12,
        &&&
        c_5
        &= - c_2,
        &&&
        d_1
        &= - c_1,
        &&&
        d_2 
        &= \tfrac12 + c_1 - c_2,
        &&&
        d_3
        &= - \tfrac12 - c_1 + c_2 - c_3,
    \end{aligned}
        \\[1mm]
    &\begin{aligned}
        f_2
        &= - \kappa,
        &&&
        h_1
        &= - g_1,
        &&&
        h_2 
        &= - g_2.
    \end{aligned}
\end{align}
The remaining coefficients of the symmetry can be easily expressed via $\beta$ from the compatibility conditions $u_{t \, \tau} - u_{\tau \, t} =0$ and $v_{t \, \tau} - v_{\tau \, t} =0 $  as follows
\begin{align}
        &\begin{aligned}
            a_1
            &= \tfrac{1}{24} \, (\beta - 2) \, (\beta^2 - 1) \beta,
            &&&
            a_2
            &= - \tfrac{1}{6} \, (\beta^2 - 4) \, (\beta - 1) \, \beta,
            &&&
            a_3
            &= \tfrac{1}{4} \, (\beta^2 - 4) \, (\beta^2 - 1),
        \end{aligned}
        \\[2mm]
        &\begin{aligned}
            a_4
            &= - \tfrac{1}{6} \, (\beta^2 - 4) \, \beta \, (\beta + 1),
            &&&
            c_1
            &= - \tfrac{1}{12} \, (\beta - 2) \, (\beta - 1) \, (\beta + 3),
            &&&
            c_2
            &= \tfrac{1}{6} \, (\beta^2 - 4) \, \beta,
        \end{aligned}
        \\[2mm]
        &\begin{aligned}
            c_3 
            &= 0,
            &&&
            e_1
            &= - \tfrac12 \kappa,
            &&&
            f_1
            &= \kappa \, (\beta - 1),
            &&&
            g_2
            &= \tfrac{1}{2} \, (2 - \beta - 4 g_1).
        \end{aligned}
\end{align}
Substituting these formulas into the compatibility conditions, we arrive at
\begin{align}
        &&
        \beta \, (\beta^2 - 1) \, (\beta^2 - 4)
        &= 0,
        &&&
        \beta \,\Big( (\beta - 2) \, (\beta - 1) 
        - 12 \,   g_1\Big)
        &= 0.
        &&
\end{align}
Thus, we proved 
\begin{proposition}
The symmetry \eqref{eq:ausymP2} with {$k = 1$} admits a non-abelianization iff   
\begin{equation}\label{beta}
   \beta = 0, \, \pm 1, \, \pm 2. 
\end{equation}
\end{proposition}
\begin{remark}
    The class of systems of the form \eqref{eq:sysP2_nc} is invariant under the matrix transposition $\tau(u)=u^T,\quad \tau(v)=v^T$. Under this transposition the parameter $\beta$ changes as follows
    \begin{equation}
        \tau(\beta)
        = - \beta.
    \end{equation}
    Therefore, there exist three non-equivalent systems with $\beta = - 2, - 1, 0$. The Hamiltonian non-commutative system {\rm(}see {\rm\cite{Kawakami_2015,bobrova2022classification})} corresponds to $\beta = 0$.
\end{remark}

The same set \eqref{beta} of admissible values of the parameter $\beta$ was obtained in \cite{Adler_Sokolov_2020_1}, where matrix \PPainleve-2 systems satisfying the \PPainleve-Kovalevskaya test were investigated. Notice that in the approach related to the test it is crucial that $u$ and $v$ are matrices (of an arbitrary size), while the above considerations can be generalized to the non-abelian case without any changes.

\subsubsection{Non-abelianization of Lax representations}
\label{sec:P2Lax}

The non-abelianization of the Lax representation \eqref{eq:LM} gives rise to the Lax pair 
\begin{align}
    \label{eq:LM1}
    L (\lambda)
    &= 
    \begin{pmatrix}
    -4 & 0 \\ 0 & 4
    \end{pmatrix}
    \lambda^2
    + 
    \begin{pmatrix}
    0 & 4 u
    \\ 
    4 u & 0
    \end{pmatrix}
    \lambda
    + 
    \begin{pmatrix}
    2 u^2 + z & - 2 u^2 + 2 v - z
    \\ 
    2 u^2 - 2 v + z & - 2 u^2 - z
    \end{pmatrix}
    + 
    \begin{pmatrix}
    0 & \kappa
    \\ 
    \kappa & 0
    \end{pmatrix}
    \lambda^{-1}
    ,
    \\[2mm]
    M (\lambda)
    &= 
    \begin{pmatrix}
    1 & 0 \\ 0 & -1
    \end{pmatrix}
    \lambda 
    + 
    \begin{pmatrix}
    - \beta \, u & - u
    \\ 
    - u & - \beta \, u
    \end{pmatrix}.
\end{align}
for auxiliary autonomous system \eqref{eq:sysP2_nc}. It turns out that the non-abelian symmetry described above has the Lax representation with the same $L$-operator and
\begin{align}
    \label{eq:ncLM1}
    M_1 (\lambda)
    = 
    \begin{pmatrix}
    [u, v] & 0 \\ 0 & [u, v]
    \end{pmatrix}
    \lambda^2
    + \tfrac12
    \begin{pmatrix}
    \Omega_1 
    & \LieBrackets{u, [u, v]} 
    \\ 
    - \LieBrackets{u, [u, v]} &
    - \Omega_1 
    \end{pmatrix}
    \lambda 
    + \tfrac12
    \begin{pmatrix}
    \Omega_2 + 2 \Omega_3 & \Omega_2
    \\
    \Omega_2 & \Omega_2 + 2 \Omega_3
    \end{pmatrix},
\end{align}
where
\begin{align}
    \Omega_1 
    &= - u^2 v - v u^2 + v^2 - 2 \kappa u - z v
    ,
    \\[2mm]
    \Omega_2
    &= u^2 v u + u v u^2
    - u v^2 + v u v - v^2 u
    + 2 \kappa u^2 + \tfrac12 z (u v + v u),
    \\[2mm]
    \Omega_3
    &= 
    - a_1 u^3 v 
    - (a_1 + a_2) u^2 v u 
    - (a_1 + a_2 + a_3) u v u^2 
    + \brackets{1 - \sum a_i} v u^3
    - c_1 u v^2 
    + f_1 u^2 
    \\
    &\quad\,
    + \brackets{- \tfrac12 - c_2} v u v 
    + (- c_1 - c_3) v^2 u 
    + \brackets{- e_1 - \tfrac12 \kappa} v
    + z \brackets{
    - g_1 u v + \brackets{\tfrac12 - g_1 - g_2} v u 
    }.
\end{align}
The coefficients of $\Omega_3$ are expressed through $\beta$ in Subsection \ref{Sec202}.
In the commutative case, the operator $M_1$ coincides with $H M.$ The relation 
$\left[ \partial_t - M, \partial_{t_{1}} - M_1\right] = 0$ is equivalent to  condition \eqref{beta}.

\section{The autonomous systems of \PPainleve-6 type}\label{P6}

In the \PPainleve-6 case, we follow the line of the previous section. However, the computations become extremely laborious. In particular, the ansatz for non-abelian symmetry contains about 1000 undetermined coefficients. 

The following trick turns out to be very helpful: we succeeded to find the leading coefficients of the non-abelian system investigating separately the homogeneous system generated by the leading parts.  

In the scalar case, the auxiliary autonomous system has the form
\begin{gather} \label{eq:scalP6sys}
\hspace{-6mm}
   \left\{
   \begin{array}{lcl}
       \, \dfrac{d u}{d t} 
        &=& 2 u^3 v 
        - 2 u^2 v 
        - \kappa_1 u^2 
        + \kappa_2 u
        + z \brackets{
        - 2 u^2 v 
        + 2 u v 
        + \kappa_4 u 
        + (\kappa_1 - \kappa_2 - \kappa_4)
        },
        \\[3mm]
       \, \dfrac{d v}{d t}
        &=& - 3 u^2 v^2 + 2 u v^2 + 2 \kappa_1 u v - \kappa_2 v + \kappa_3
        + z \brackets{
     2 u v^2 - v^2 - \kappa_4 v
     },
\end{array}
\right.
\hspace{-8mm}
\end{gather}
where  $z \in \mathbb{C}$ and $\kappa_i$ are arbitrary parameters. The Hamiltonian $H$ for \eqref{eq:scalP6sys} is given by 
\begin{gather}\label{scalHam}
   H 
   = H_1 
    + z \, H_2 
    ,
\end{gather}
with
\begin{align}
    H_1
    &= 
    u^3 v^2
    - u^2 v^2
    - \kappa_1 u^2 v
    + \kappa_2 u v
    - \kappa_3 u,
    &
    H_2
    &= 
    - u^2 v^2
    + u v^2
    + \kappa_4 u v
    + (\kappa_1 - \kappa_2 - \kappa_4) v.
\end{align}

The non-abelianization of \eqref{eq:scalP6sys} is a system of the form \eqref{eq:noz}, where $P_i$ and $Q_i$ are non-commutative polynomials 
given~by
\begin{align} \label{eq:P1form}
\begin{aligned}
\begin{aligned}
    &\begin{aligned}
    P_1 (u, v)
    = a_1 u^3 v 
    + a_2 u^2 v u 
    + a_3 u v u^2 
    + (2 - a_1 - a_2 - a_3) v u^3
    + c_1 u^2 v 
    \\
    + \, (- 2 - c_1 - c_2) u v u 
    + c_2 v u^2
    - \kappa_1 u^2 
    + \kappa_2 u,
    \end{aligned}
    \\[1mm]
    &\begin{aligned}
    Q_1 (u, v)
    = f_1 u^2 v 
    + (- 2 - f_1 - f_2) u v u 
    + f_2 v u^2
    + h_1 u v 
    + (2 - h_1) v u 
    + \kappa_4 u 
    \\
    + \, (\kappa_1 - \kappa_2 - \kappa_4),
    \end{aligned}
    \\[2mm]
    &\begin{aligned}
    P_2 (u, v)
    = b_1 u^2 v^2 
    + b_2 u v u v 
    + b_3 u v^2 u 
    + b_4 v u^2 v
    + b_5 v u v u
    + \brackets{- 3 - \sum b_i} v^2 u^2
    \\
    + \, d_1 u v^2
    + (2 - d_1 - d_2) v u v 
    + d_2 v^2 u 
    + e_1 u v 
    + (2 \kappa_1 - e_1) v u
    - \kappa_2 v 
    + \kappa_3,
    \end{aligned}
    \\[1mm]
    &\begin{aligned}
    Q_2 (u, v)
    = g_1 u v^2
    + (2 - g_1 - g_2) v u v 
    + g_2 v^2 u 
    - v^2 - \kappa_4 v
    \end{aligned}
\end{aligned}
\end{aligned}
\end{align}
where $a_i$, $b_i$, $c_i$, $d_i$, $f_i$, $g_i$, $h_1$, and $e_1$ belong to $\mathbb{C}$.

It is clear that if system \eqref{eq:noz}, \eqref{eq:P1form} possesses integrals, symmetries and Lax representation, then their leading parts define the corresponding objects for its leading homogeneous part. For this reason, we first perform a non-abelianization of the homogeneous part of the scalar autonomous system of \PPainleve-6 type.
 
\subsection{Homogeneous non-abelian systems of \PPainleve-6 type}
The homogeneous principal part\footnote{Here and below we denote $t$-derivatives by $'$.},
\begin{gather} 
\label{eq:scalP6sysAux}
    \left\{
    \begin{array}{lcl}
         u' 
         &=& 2 u^3 v,
         \\[2mm]
         v'
         &=& - 3 u^2 v^2
         ,
    \end{array}
    \right.
\end{gather}
of \eqref{eq:scalP6sys}  
possesses the first integral
\begin{equation}
    I = u^3 v^2.
\end{equation}
For any $N$ the following system
\begin{equation}\label{scsym9}
\left\{
\begin{array}{lcl}
u_{\tau} &=& u' \, I^N, \\[2mm]
v_{\tau} &=& v' \, I^N
\end{array}
\right.
\end{equation}
is an infinitesimal symmetry of \eqref{eq:scalP6sysAux}. In addition to \eqref{scsym9}, the system \eqref{eq:scalP6sysAux} has symmetries
\begin{equation}\label{scsym6}
\left\{
\begin{array}{lcl}
u_{\tau} &=& u \,I^N, \\[2mm]
v_{\tau} &=& -2 v\,I^N,
\end{array}
\right.
\end{equation}
which are related to the homogeneity of the system. The degrees of the simplest symmetries are equal to 6 and 9.
We are going to find all non-abelianizations of system \eqref{eq:scalP6sysAux} that have non-abelian symmetries of degree 6 and 9. 

The non-abelianization of \eqref{eq:scalP6sysAux} is given by  
\begin{align} \label{eq:matP6syshomGen}
    \left\{
    \begin{array}{lcr}
        u' 
        &=& a_1\, u^3 v + a_2\, u^2 v u+a_3\, u v u^2 
        + (2 - a_1 - a_2 - a_3)\, v u^3 
        ,
        \\[2mm]
        v'
        &=& b_1 \, u^2 v^2 
        + b_2 \, u v u v 
        + b_3 \, u v^2 u
        + b_4 \, v u^2 v 
        + b_5 \, v u v u
        \hspace{7mm}
        \\[1mm]
        &&
        + \, (-3   - b_1 - b_2  - b_3  - b_4 - b_5) \, v^2 u^2
         .
    \end{array}
    \right.
\end{align}
As in the \PPainleve-2 case (see Subsection \ref{sec:P2int}), we are going to generalize to the non-abelian case several first integrals. It allows us to reduce the number of unknown coefficients in \eqref{eq:matP6syshomGen}.
\begin{lemma}\label{lemma1} If the integrals $I$ and $I^2$ admit a non-abelianization then the system \eqref{eq:matP6syshomGen} has the form 
\begin{align} \label{eq:homP6_simplified}
    \left\{
    \begin{array}{lcr}
        u' 
        &=& a_1\, u^3 v + a_2\, u^2 v u+a_3\, u v u^2 
        + (2 - a_1 - a_2 - a_3)\, v u^3 
        ,
        \\[2mm]
        v'
        &=& - a_1 \, u^2 v^2 
        + (- a_1 - a_2) \, u v u v 
        + (1 - a_2 - a_3) \, v u^2 v
        \hspace{5mm}
        \\[1mm]
        &&
        + \, (- 2 + a_1 + a_2) \, v u v u
        + (- 2 + a_1 + a_2 + a_3) \, v^2 u^2
         .
    \end{array}
    \right.
\end{align}
\end{lemma}

Following the line of Subsection \ref{Sec202}, we specify the remaining parameters in \eqref{eq:homP6_simplified} by the non-abelianization of symmetries.
\begin{proposition}
\label{thm:hompart}
    If system \eqref{eq:homP6_simplified} admits a non-abelianization of symmetries of degrees 6 and 9 then it has one of the following forms \text{\rm(}up to the transposition $u\to u^T, \, v\to v^T$\text{\rm)}
    \begin{align}
        \text{\rm{\textbf{Case 1:}}}
        &&&
        \begin{aligned}
        a_1 &= a,
        &&&
        a_2 &= 2 - a,
        &&&
        a_3 &= 0,
        \end{aligned}
        \\[1mm]
        &&&
        \left\{
        \begin{array}{lcl}
             u'
             &=& a \, u^3 v 
             + (2 - a) \, u^2 v u,  
             \\[2mm]
             v'
             &=& - a \, u^2 v^2 
             - 2 u v u v 
             + (- 1 + a) \, v u^2 v; 
        \end{array}
        \right.
        \\[2mm]
        \text{\rm{\textbf{Case 2:}}}
        &&&
        \begin{aligned}
        a_1 &= 0,
        &&&
        a_2 &= a,
        &&&
        a_3 &= 2 - a,
        \end{aligned}
        \\[1mm]
        &&&
        \left\{
        \begin{array}{lcl}
             u'
             &=& a \, u^2 v u 
             + (2 - a) \, u v u^2, 
             \\[2mm]
             v'
             &=& - a \, u v u v 
             - v u^2 v 
             + (- 2 + a) v u v u
             ,
        \end{array}
        \right.
    \end{align}
with $a =-1,0,1,2$ and $a =-1,0,1,2,3$, respectively.
\end{proposition}

\begin{proof}
In the proof, we first find relations between the coefficients of system \eqref{eq:homP6_simplified} which is equivalent to the existence of a non-abelianization of the symmetry of degree 6.

The non-abelianization ansatz for symmetry \eqref{scsym6} with $N=1$ is  obvious, and we do not present it here.
Let us denote the unknown coefficients of the monomials $u^2 v u v u$, $u v u^2 v u$, $u v u v u^2$ in the ansatz by $c_1$, $c_2$, $c_3$, respectively. The compatibility symmetry condition produces a large system of non-linear algebraic equations. Eliminating variables, we arrive at an autonomous subsystem for the coefficients $a_1$, $a_2$, $a_3$, $c_1$, $c_2$, $c_3$. One of its equations is $a_1 \, a_3 = 0$ which leads to two different cases.

If $a_1 = 0$, then the subsystem is equivalent to
\begin{align}
    &
    \begin{aligned}
    \left(a_2+a_3-2\right) c_2
    &= 0,
    &
    \left(a_2-2\right) a_2
    - \left(a_3-2\right) a_3 c_2
    &= 0,
    &
    a_3 \left(\left(a_3-2\right) c_2+a_2\right)
    &= 0,
    \end{aligned}
    \\[2mm]
    &
    \begin{aligned}
    \left(a_3-2\right) a_3 
    \left(c_2-1\right) c_2
    &= 0,
    &
    \left(a_3-1\right) c_3
    &= 0,
    &
    c_3 \left(a_2-c_2\right)
    &= 0,
    &
    \left(c_2-1\right) c_2 c_3
    &= 0.
    \end{aligned}
\end{align}
This system has the following solution
\begin{align}
    &&
    a_3
    &= 2 - a,
    &
    a_2
    &= a,
    &
    c_2
    &= 1,
    &
    c_3 
    &= 0,
    &&
\end{align}
where $a$ is an arbitrary parameter.

In the case $a_3 = 0$, we obtain the following system
\begin{align}
    &
    \begin{aligned}
    \left(a_1-a_2\right) \left(a_1+a_2-2\right)
    &= 0,
    &
    a_2 \left(a_1+a_2-2\right)
    &= 0,
    &
    \left(a_1-1\right) c_1
    &= 0,
    \end{aligned}
    \\[2mm]
    &
    \begin{aligned}
    \left(a_2-1\right) c_1
    &= 0,
    &
    a_1 c_2
    &= 0,
    &
    c_1 c_2
    &= 0,
    &
    \left(a_2-2\right) c_2
    &= 0,
    &
    c_3 
    &= 0.
    \end{aligned}
\end{align}
The first two equations give 
\begin{align}
    &&
    a_2
    &= 2 - a,
    &
    a_1
    &= a.
    &&
\end{align}

To specify an arbitrary parameter $a$, we use a non-abelian symmetry of degree 9. Let us consider the compatibility condition and denote the unknown coefficients related to the symmetry that cannot be determined uniquely from this condition by $d_1$ and $d_2$. Then, in Case \textbf{1}, these monomials are $d_1 \, u^3 v u^2 v u v$, $d_2 \, u^3 v u v u^2 v$ and one can obtain a system that is equivalent to
\begin{align}
    &
    \begin{aligned}
        d_1-(a-1) d_2
        &= 0,
        &&&
        \left(d_2-1\right) d_2
        &= 0,
        &&&
        d_1 \left(d_2-1\right)
        &= 0,
    \end{aligned}
    \\[2mm]
    &
    \begin{aligned}
        a^2+a-4 d_1-2 d_2
        &= 0,
        &&&
        \left(d_1-1\right) d_1
        &= 0,
        &&&
        (a-2) d_1
        &= 0.
    \end{aligned}
\end{align}
So, the solution is one of the following:
\begin{align}
    &
    \begin{aligned}
        a
        &= -1,
        &
        d_1
        &= 0,
        &
        d_2
        &= 0;
    \end{aligned}
    &
    \begin{aligned}
        a
        &= 0,
        &
        d_1
        &= 0,
        &
        d_2
        &= 0;
    \end{aligned}
    \\[2mm]
    &
    \begin{aligned}
        a
        &= 1,
        &
        d_1
        &= 0,
        &
        d_2
        &= 1;
    \end{aligned}
    &
    \begin{aligned}
        a
        &= 2,
        &
        d_1
        &= 1,
        &
        d_2
        &= 1.
    \end{aligned}
\end{align}
Similarly, in Case \textbf{2}, a reduced system can be written as
\begin{align}
    &
    \begin{aligned}
    \left(d_2-1\right) d_2
    &= 0,
    &&&
    d_1 \left(d_1+1\right)
    &= 0,
    &&&
    d_2 \left(a-d_1\right)
    &= 0,
    \end{aligned}
    \\[2mm]
    &
    \begin{aligned}
    d_1 \left(a+4 d_2-3\right)
    &= 0,
    &&&
    (a-3) a+2 \left(d_1+\left(d_1-1\right) d_2+1\right)
    &= 0,
    \end{aligned}
\end{align}
where $d_1$ and $d_2$ are coefficients of $u v u^2 v u^2 v u$ and $u v u v u^2 v u^2$, respectively. The system possesses one of the following solutions:
\begin{gather}
    \begin{aligned}
        a
        &= -1,
        &
        d_1
        &= -1,
        &
        d_2
        &= 1;
        &&&
        a
        &= 0,
        &
        d_1
        &= 0,
        &
        d_2
        &= 1;
        &&&
        a
        &= 1,
        &
        d_1
        &= 0,
        &
        d_2
        &= 0;
    \end{aligned}
    \\[2mm]
    \begin{aligned}
        a
        &= 2,
        &
        d_1
        &= 0,
        &
        d_2
        &= 0;
        &&&
        a
        &= 3,
        &
        d_1
        &= -1,
        &
        d_2
        &= 0.
    \end{aligned}
\end{gather}
Therefore, we arrive at the statement of the proposition.
\end{proof}

\begin{remark} 
Note that Case {\rm\textbf{1}}, $a = 0$ coincides with Case {\rm\textbf{2}}, $a = 2$.  
\end{remark}

\begin{remark} 
The matrix transposition $\tau: \, u\to u^T, \, v\to v^T$  preserves the class of systems \eqref{eq:noz}, \eqref{eq:P1form} and maps integrable systems to integrable.  
In particular, systems related to Case {\rm\textbf{2}} from Proposition {\rm\ref{thm:hompart}} with $a=0$ and $a=2$ as well as with $a=-1$ and $a = 3$ are dual with respect to the transposition. Sometimes we formulate classification statements up to the transposition.  
\end{remark}

\subsection{General case}

Considering system \eqref{eq:noz}, \eqref{eq:P1form}, we obtain the following generalization of Lemma \ref{lemma1}:
\begin{proposition}
\label{thm:intP6_nc}
An autonomous system of the form \eqref{eq:noz}, \eqref{eq:P1form} admits a non-abelianization of the integrals $H_6$ and $H_6^2$ iff the polynomials $P_i$ and $Q_i$ have the form
\begin{align}
    \label{eq:P6sys_rhs1}
    &\begin{aligned}
    P_1 (u, v)
    &= a_1 \, u^3 v 
    + a_2 \, u^2 v u + a_3 \, u v u^2 
    + (2 - a_1 - a_2 - a_3) \, v u^3
    \\
    & \qquad
    + \, c_1 \, u^2 v + (- 2 - c_1 - c_2) \, u v u + c_2 \, v u^2
    - \kappa_1 \, u^2 + \kappa_2 \, u,
    \\
    Q_1 (u, v)
    &= f_1 \, u^2 v 
    + (- 2 - f_1 - f_2) \, u v u + f_2 \, v u^2
    \\
    & \qquad
    + \, h_1 \, u v + (2 - h_1) \, v u 
    + \kappa_4 u
    + \brackets{\kappa_1 - \kappa_2 - \kappa_4},
    \end{aligned}
    \\[2mm]
    \label{eq:P6sys_rhs2}
    &\begin{aligned}
    P_2 (u, v)
    &= - a_1 \, u^2 v^2 
    + (- a_1 - a_2) \, u v u v 
    + (1 - a_2 - a_3) \, v u^2 v 
    \\
    & \qquad
    + \, (- 2 + a_1 + a_2) \, v u v u
    + (-2 + a_1 + a_2 + a_3) \, v^2 u^2
    \\
    & \qquad
    - \, c_1 \, u v^2 + (2 + c_1 + c_2) \, v u v 
    - c_2 \, v^2 u
    \\
    & \qquad
    + \, e_1 \, u v
    + (2 \kappa_1 - e_1) \, v u
    - \kappa_2 \, v
    + \kappa_3,
    \\
    Q_2 (u, v)
    &= - f_1 \, u v^2 
    + (2 + f_1 + f_2) \, v u v - f_2 \, v^2 u
    - v^2 - \kappa_4 \, v.
    \end{aligned}
\end{align}
\vspace{-3mm}
\qed
\end{proposition}

As it was mentioned at the beginning of Section \ref{P6}, if system \eqref{eq:P1form} satisfy Assumption~\ref{assumpt1}, then its leading homogeneous part also satisfies it. Thus, Proposition \ref{thm:hompart} separates the classification of systems \eqref{eq:noz}, \eqref{eq:P6sys_rhs1}, \eqref{eq:P6sys_rhs2} into 9 different cases. We consider each of these cases independently\footnote{In fact, it suffices to consider 5 cases that are not equivalent with respect to the transposition $\tau$.}.

An inhomogeneous scalar \PPainleve-6\,system does not have symmetry of order 6. Therefore, we consider non-abelianization of the symmetry of order 9.

Just as in Subsection \ref{Sec202}, we refine the form of the symmetry, after which about 300 indeterminate coefficients remain in the ansatz for the symmetry. The compatibility condition of the system and symmetry leads to a huge overdetermined system of non-linear algebraic equations with respect to the remaining nine coefficients in the equation and the symmetry coefficients.  Despite its size, the system is easily solved using the computer algebra system ``Mathematica''. 

\begin{theorem}
\label{thm:sysP6_nc}
A system \eqref{eq:noz} with  \eqref{eq:P6sys_rhs1}, \eqref{eq:P6sys_rhs2}, whose leading part is decribed by Proposition {\rm\ref{thm:hompart}},   admits a non-abelianization of the symmetry of degree {\rm9} iff it belong to one of the following lists
\begin{enumerate}
    \item[\rm{i)}] 
    Appendix A.1 from the paper {\rm\cite{bobrova2022classification}},
    
    \item[\rm{ii)}] 
    Appendix A.1 from the paper {\rm\cite{bobrova2022non}},
    
    \item[\rm{iii)}] 
    Appendix \ref{sec:sysintlistP6} in this paper {\rm(}up to the transposition $\tau${\rm).}
    \qed
\end{enumerate}
\end{theorem}
 
\section{Autonomous systems of \texorpdfstring{$\PV - \PI$}{P5 - P1} types}\label{P5-P1}
\subsection{Systems of \texorpdfstring{$\PV$}{P5} type}
The scalar auxiliary autonomous system of the \PPainleve-5 type is defined by the Hamiltonian $H$ of the form
\begin{equation}
    \label{eq:ham5}
     H
    = u^3 v^2 
    - 2 u^2 v^2 
    + v^2 
    - \kappa_1 u^2 v
    + (\kappa_1 + \kappa_2) u v
    - \kappa_2 v
    - \kappa_3 u
    + \kappa_4 z u v.
\end{equation}
A non-abelian generalization of this system can be written as
\begin{align}
    \label{eq:sysP5_nc}
    \left\{
    \begin{array}{lcr}
         \dfrac{du}{dt}
        &=& a_1 u^3 v 
        + a_2 u^2 v u 
        + a_3 u v u^2 
        + \brackets{
        2 - \sum a_i
        } v u^3
        + c_1 u^2 v 
        + (- 4 - c_1 - c_2) u v u 
        \hspace{2mm}
        \\[1mm]
        &&
        + \, c_2 v u^2
        - \kappa_1 u^2
        + e_1 u v + (2 - e_1) v u 
        + (\kappa_1 + \kappa_2) u
        - \kappa_2
        + \kappa_4 z u
        ,
        \\[2mm]
        \dfrac{dv}{dt}
        &=& b_1 u^2 v^2 
        + b_2 u v u v 
        + b_3 u v^2 u
        + b_4 v u^2 v
        + b_5 v u v u
        + \brackets{
        - 3 - \sum b_i
        } v^2 u^2
        \hspace{1.2cm}
        \\[1mm]
        &&
        + \, d_1 u v^2 
        + (4 - d_1 - d_2) v u v 
        - d_2 v^2 u 
        - v^2 
        + f_1 u v + (2 \kappa_1 - f_1) v u
        \\[1mm]
        &&
        - \, (\kappa_1 + \kappa_2) v
        + \kappa_3
        - \kappa_4 z v
        ,
    \end{array}
    \right.
\end{align}
where $a_i$, $b_i$, $c_i$, $d_i$, $e_1$, $f_1$, $\kappa_i \in \mathbb{C}$. 

\begin{proposition}
\label{thm:intP5_nc}
\phantom{}
\begin{itemize}
    \item[\rm{a)}]
    A system of the form \eqref{eq:noz} admits a non-abelianization of integrals $H$ and $H^2$ iff
\begin{gather}
\begin{aligned}
    b_1
    &= - a_1,
    &&&
    b_2
    &= - a_1 - a_2,
    &&&
    b_3 
    &= 0,
    &&&
    b_4
    &= 1 - a_2 - a_3,
    &&&
    b_5
    &= - 2 + a_1 + a_2,
\end{aligned}
\\
\begin{aligned}
    &&
    d_1
    &= - c_1,
    &&&
    d_2
    &= - c_2.
    &&
\end{aligned}
\end{gather}

\item[\rm{b)}]
If the system additionally admits a non-abelianiazation of the symmetry \eqref{scsymgen} of degree 9, then it belongs to one of the following lists
\begin{enumerate}
    \item[\rm{i)}] 
    Appendix A.2 from {\rm\cite{bobrova2022classification}},
    
    \item[\rm{ii)}] 
    Appendix A.2 from {\rm\cite{bobrova2022non}},
    
    \item[\rm{iii)}] 
    Appendix \ref{sec:sysintlistP5} in this paper {\rm(}up to the transposition $\tau${\rm).}
    \qed
\end{enumerate}
\end{itemize}
\end{proposition}

\subsection{Systems of \texorpdfstring{$\PIV$}{P4} type}
The scalar Hamiltonian for an autonomous $\PIV$ system is given by
\begin{align}
    \label{eq:ham4}
    H
    &= u v^2 - u^2 v
    + \kappa_2 v
    - \kappa_3 u 
    - 2 z u v.
\end{align}
Then, the general non-abelian autonomous system can be written in the form
\begin{align} \label{eq:sysP4_nc}
    \left\{
    \begin{array}{lcl}
       \dfrac{d u}{d t}
         &=& - u^2 
         + 2 u v + \alpha \, [u, v]
         - 2 z u
         + \kappa_2,
         \\[3mm]
        \dfrac{d v}{d t}
         &=& - v^2
         + 2 v u + \beta \, [v, u]
         + 2 z v
         + \kappa_3
    \end{array}
    \right.
\end{align}
with the parameters $\alpha$, $\beta\in \C$ to be defined. 

\begin{proposition}
\label{thm:sysP4_nc}
\phantom{}
\begin{itemize}
    \item[\rm{a)}]
    For any parameters $\alpha$, $\beta$ the system   \eqref{eq:noz} admits a non-abelianization of the first integrals $H$ and $H^2$.

    \item[\rm{b)}] The autonomous system possesses a non-abelian symmetry of degree $5$ iff it belongs to one of the following lists
    \begin{enumerate}
        \item[\rm{i)}] 
        Appendix A.3 from {\rm\cite{bobrova2022classification}},
        
        \item[\rm{ii)}] 
        Appendix A.3 from {\rm\cite{bobrova2022non}},
        
        \item[\rm{iii)}] 
        Appendix \ref{sec:sysintlistP4} {\rm(}up to the transposition~$\tau${\rm)}.
    \qed
    \end{enumerate}
\end{itemize} 
\end{proposition}

\subsection{\texorpdfstring{$\PIII$}{P3} type systems}
In the $\PIII$ case, we deal with the \PPainleve-$3^{\prime}$ version of the scalar system.  The Hamiltonian of the corresponding autonomous system is given by:
\begin{equation}
    \label{eq:ham3}
    H
    = u^2 v^2 
    + \kappa_2 u^2 v
    + \kappa_1 u v 
    + \kappa_3 u 
    + \kappa_4 z v.
\end{equation}
A non-abelian generalization of the system can be written as 
\begin{align}
    \label{eq:sysP3_nc}
    \left\{
    \begin{array}{lclcl}
          \dfrac{du}{dt}
         &=& a_1 u^2 v 
         + (2 - a_1 - a_2) u v u
         + a_2 v u^2
         + \kappa_1 u + \kappa_2 u^2
         + \kappa_4 z
         ,
         \\[3mm]
         \dfrac{dv}{dt}
         &=& b_1 u v^2 
         - (2 + b_1 + b_2) v u v
         + b_2 v^2 u
         - \kappa_1 v
         + c_1 u v
         + (- 2 \kappa_2 - c_1) v u
         - \kappa_3,
    \end{array}
    \right.
\end{align}
where $a_i$, $b_i$, $c_1$, $\kappa_i \in \mathbb{C}$. 
\begin{proposition}
\label{thm:intP3_nc}
\phantom{}
\begin{itemize}
    \item[\rm{a)}] 
A system of the form \eqref{eq:sysP3_nc} admits a non-abelianization of integrals $H$ and $H^2$ iff
 \begin{align}
    &&
    &&
    b_1
    &= - a_1,
    &&&
    b_2
    &= - a_2.
    &&
    &&
\end{align}

\item[\rm{b)}]  If the system additionally admits a non-abelianiazation of the symmetry \eqref{scsymgen} of degree 7, then it belongs to one of the following lists
\begin{enumerate}
    \item[\rm{i)}] 
    Appendix A.4 from {\rm\cite{bobrova2022classification}},
    
    \item[\rm{ii)}] 
    Appendix A.4 from {\rm\cite{bobrova2022non}},
    
    \item[\rm{iii)}] 
    Appendix \ref{sec:sysintlistP3} {\rm(}up to the transposition~$\tau${\rm)}.
    \qed
\end{enumerate}
\end{itemize}
\end{proposition}

\subsection{Systems of \texorpdfstring{$\PII$}{P2} type}
These systems were considered in Subsection \ref{sec:P2_case} in details. There are three systems non-equivalent with respect to the transposition, which correspond to results of the paper \cite{Adler_Sokolov_2020_1} (see also \cite{Retakh_Rubtsov_2010}).

\subsection{\texorpdfstring{$\PI$}{P1} type system}\label{P1}
In the scalar case, an autonomous system has the form  
\begin{align}
    \label{eq:sysP1}
    \left\{
    \begin{array}{lcl}
         \dfrac{du}{dt}
         &=& v,
         \\[3mm]
         \dfrac{dv}{dt}
         &=& 6 u^2 + z.
    \end{array}
    \right.
\end{align}
For a non-abelianization of this system we can just replace the variables $u$ and $v$ with non-commutative ones. It turns out that this analog admits Assumption \ref{assumpt1}.
Probably, the corresponding matrix equation \PPainleve-1 first time appeared in the paper \cite{Balandin_Sokolov_1998}.

\section{Non-abelian \Painleve systems}\label{pensyst}

Formally, the restoring of a \Painleve system of the form \eqref{eq:fullansatz} corresponding to any autonomous system \eqref{eq:noz} described in Sections \ref{P6} and \ref{P5-P1} is very simple. We have to replace the time $t$ with $z$ and choose $f(z)=z (z-1)$ in the $\PVI$ case, $f(z)=z$ for $\PV$ and $\PIII$ systems, and $f(z)=1$ otherwise.

The question arises whether these systems are integrable. It is not known whether all of them satisfy the \PPainleve \, property. However,
\begin{itemize}
    \item all known matrix \PPainleve \, systems can be obtained  in this way;
    \item all these systems possess an isomonodromic Lax representation;
    \item systems of different types are related by limiting transitions in the same way as in the scalar case and can be obtained by degenerations of the $\PVI$ systems.
\end{itemize}

\subsection{Transformation groups and their orbits}\label{Sec5.1}

\subsubsection{\texorpdfstring{$\PVI$}{P6} case} 
\label{sec:trgroup_P6}
The scalar \PPainleve-6 system admits the linear B\"acklund transformations  
\begin{align}
    \label{eq:scalP6sym1}
    &&
    r_1:&
    &
    \brackets{
    z, u, v
    }
    &\mapsto \brackets{
    1 - z, \, 
    1 - u, \,
    - v
    },
    &&&
    r_2:&
    &
    \brackets{
    z, u, v
    }
    &\mapsto \brackets{
    z^{-1}, \,
    z^{-1} u, \,
    z v
    },
    &&
\end{align}
which change the parameters as 
\begin{align}
    \label{eq:scalP6sym1k}
    &\begin{aligned}
    r_1:&
    &
    \brackets{
    \kappa_1, 
    \kappa_2, 
    \kappa_3, 
    \kappa_4
    }
    &\mapsto \brackets{
    \kappa_1, \,
    2 \kappa_1 - \kappa_2 - \kappa_4, \,
    \kappa_3, \,
    \kappa_4
    },
    \end{aligned}
    \\[2mm]
    \label{eq:scalP6sym2k}
    &\begin{aligned}
    r_2:&
    &
    \brackets{
    \kappa_1, 
    \kappa_2, 
    \kappa_3, 
    \kappa_4
    }
    &\mapsto \brackets{
    \kappa_1, \,
    \kappa_4 - 1, \,
    \kappa_3, \,
    \kappa_2 + 1
    }.
    \end{aligned}
\end{align}
These involutions generate a group isomorphic to $S_3$. 

In the non-abelian case the transposition  $\tau$ together with the above transformations act on the set of integrable non-abelian \PPainleve-6 systems.  The transposition $\tau$ changes the parameters in \eqref{eq:P6sys_rhs1}, \eqref{eq:P6sys_rhs2}  as follows:
\begin{multline}
    \tau \brackets{
    a_1, \,
    a_2, \,
    a_3, \,
    c_1, \,
    c_2, \,
    e_1, \,
    f_1, \,
    f_2, \,
    h_1
    }
    \\
    = \brackets{
    2 - a_1 - a_2 - a_3, \,
    a_3, \,
    a_2, \,
    c_2, \,
    c_1, \,
    2 \kappa_1 - e_1, \,
    f_2, \,
    f_1, \,
    2 - h_1
    }.
\end{multline}
Systems \ref{eq:P6_1} -- \ref{eq:P6_16} of  the $\PVI$ type  (see  Appendix \ref{sec:sysintlistP6}) are non-equivalent with respect to  $\tau$.  The $S_3$-action on these systems is given by:
\begin{align}
    r_1 {\text{\eqref{eq:P6_1}}}
    &= \text{\ref{eq:P6_1}},
    &
    r_2 {\text{\eqref{eq:P6_1}}}
    &= \text{\ref{eq:P6_1}},
    &
    r_1 {\text{\eqref{eq:P6_2}}}
    &= \text{\ref{eq:P6_7}},
    &
    r_2 {\text{\eqref{eq:P6_2}}}
    &= \text{\ref{eq:P6_2}},
    \\[2mm]
    r_1 {\text{\eqref{eq:P6_4}}}
    &= \text{\ref{eq:P6_4}},
    &
    r_2 {\text{\eqref{eq:P6_4}}}
    &= \text{\ref{eq:P6_7}},
    &
    r_1 {\text{\eqref{eq:P6_7}}}
    &= \text{\ref{eq:P6_2}},
    &
    r_2 {\text{\eqref{eq:P6_7}}}
    &= \text{\ref{eq:P6_4}},
    \\[2mm]
    r_1 {\text{\eqref{eq:P6_8}}}
    &= \text{\ref{eq:P6_8}},
    &
    r_2 {\text{\eqref{eq:P6_8}}}
    &= \text{\ref{eq:P6_8}},
    &
    r_1 {\text{\eqref{eq:P6_17}}}
    &= \text{\ref{eq:P6_16}},
    &
    r_2 {\text{\eqref{eq:P6_17}}}
    &= \text{\ref{eq:P6_14}},
    \\[2mm]
    r_1 {\text{\eqref{eq:P6_14}}}
    &= \text{\ref{eq:P6_14}},
    &
    r_2 {\text{\eqref{eq:P6_14}}}
    &= \text{\ref{eq:P6_17}},
    &
    r_1 {\text{\eqref{eq:P6_16}}}
    &= \text{\ref{eq:P6_17}},
    &
    r_2 {\text{\eqref{eq:P6_16}}}
    &= \text{\ref{eq:P6_16}}.
\end{align}
Transformations $r_1$, $r_2$ taking together with the transposition $\tau$ generate a group isomorphic to $S_3 \times \mathbb{Z}_2$. The action of this group has the following four orbits:
\begin{align*}
    \text{\bf Orbit 1} 
    &= \left\{ 
        \text{\ref{eq:P6_1}}, \,
        \tau\text{\eqref{eq:P6_1}}
    \right\},
    &
    \text{\bf Orbit 2} 
    &= \left\{ 
        \text{\ref{eq:P6_2}}, \,
        \tau\text{\eqref{eq:P6_2}}, \,
        \text{\ref{eq:P6_4}}, \,
        \tau\text{\eqref{eq:P6_4}}, \,
        \text{\ref{eq:P6_7}}, \,
        \tau\text{\eqref{eq:P6_7}}
    \right\},
    \\[2mm]
    \text{\bf Orbit 3} 
    &= \left\{ 
        \text{\ref{eq:P6_8}}, \,
        \tau\text{\eqref{eq:P6_8}}
    \right\},
    &
    \text{\bf Orbit 4} 
    &= \left\{ 
        \text{\ref{eq:P6_17}}, \,
        \tau\text{\eqref{eq:P6_17}}, \,
        \text{\ref{eq:P6_14}}, \,
        \tau\text{\eqref{eq:P6_14}}, \,
        \text{\ref{eq:P6_16}}, \,
        \tau\text{\eqref{eq:P6_16}}
    \right\}.
\end{align*}

Besides four non-equivalent new systems corresponding to the orbits, we have five non-equivalent systems with the generalized Okamoto integral found in \cite{bobrova2022non} and one Hamiltonian non-abelian system \cite{Kawakami_2015,bobrova2022classification}.

\subsubsection{\texorpdfstring{$\PV$}{P5} case} 
The transposition $\tau$ acts on the parameters of system \eqref{eq:sysP5_nc} as
\begin{align}
    \tau \brackets{
    a_1, \,
    a_2, \,
    a_3, \,
    c_1, \,
    c_2, \,
    e_1, \,
    f_1
    }
    &= \brackets{
    2 - a_1 - a_2 - a_3, \,
    a_3, \,
    a_2, \,
    c_2, \,
    c_1, \,
    2 - e_1, \,
    2 \kappa_1 - f_1
    }.
\end{align}
Its action on the systems from Appendix \ref{sec:sysintlistP5} forms four orbits:
\begin{align*}
    \text{\bf Orbit 1} 
    &= \left\{ 
        \text{\ref{eq:P5_1}}, \,
        \tau\text{\eqref{eq:P5_1}}
    \right\},
    &
    \text{\bf Orbit 2} 
    &= \left\{ 
        \text{\ref{eq:P5_2}}, \,
        \tau\text{\eqref{eq:P5_2}}
    \right\},
    \\[2mm]
    \text{\bf Orbit 3} 
    &= \left\{ 
        \text{\ref{eq:P5_5}}, \,
        \tau\text{\eqref{eq:P5_5}}
    \right\},
    &
    \text{\bf Orbit 4} 
    &= \left\{ 
        \text{\ref{eq:P5_9}}, \,
        \tau{\text{\eqref{eq:P5_9}}}
    \right\}.
\end{align*}
Thus we have 5 non-equivalent systems obtained in \cite{bobrova2022non}, one Hamiltonian system \cite{Kawakami_2015,bobrova2022classification} and 4 new systems found in this paper.

\subsubsection{\texorpdfstring{$\PIV$}{P4} case} 
The transposition $\tau$ and the involutions
\begin{align}
    \label{eq:s1s2P4_nc}
    s_1 \brackets{
    u, \, v, \, z
    }
    &= \brackets{
    i v, \, i u, \, - i z
    },
    &
    s_2 \brackets{
    u, \, v, \, z
    }
    &= \brackets{
    - i u, \, i (v - u - 2 z), \, - i z
    }
\end{align}
preserve a class of systems of the form \eqref{eq:sysP4_nc}, changing the parameters as
\begin{align}
    \label{eq:symP4_nc}
    \tau \brackets{
    \alpha, \, \beta
    }
    &= \brackets{
    - 2 - \alpha, \, - 2 - \beta
    },
    &
    s_1 \brackets{
    \alpha, \, \beta
    }
    &= \brackets{
    \beta, \, \alpha
    },
    &
    s_2 \brackets{
    \alpha, \, \beta
    }
    &= \brackets{
    \alpha, \, - \alpha - \beta - 3
    }. \,\,
\end{align}
There are three systems non-equivalent with respect to the group generated by $\tau$, $s_1$, and $s_2$. One of them, Hamiltonian, was found in \cite{Kawakami_2015}, the second has the Okamoto integral \cite{bobrova2022non}, and the third one is presented in Appendix \ref{sec:sysintlistP4}. All of them also arose in the paper \cite{Bobrova_Sokolov_2022}, where the \PPainleve-Kovalevskaya test was used for classification (see also \cite{adler2020}).

\subsubsection{\texorpdfstring{$\PIII$}{P3} case} 
The transposition $\tau$ acts on the parameters of system \eqref{eq:sysP3_nc} in the following way
\begin{align}
    \tau \brackets{
    a_1, \,
    a_2, \,
    c_1
    }
    &= \brackets{
    a_2, \,
    a_1, \,
    - c_1 - 2 \kappa_2
    }.
\end{align}
Appendix \ref{sec:sysintlistP3} contains two systems non-equivalent with respect to $\tau$.  Four systems with the Okamoto integral can be found in \cite{bobrova2022classification}. For one more, Hamiltonian system, see \cite{Kawakami_2015,bobrova2022classification}.

\subsubsection{\texorpdfstring{$\PII$}{P2} case}
There are three systems non-equivalent with respect to the transposition, which correspond to results of the paper \cite{Adler_Sokolov_2020_1} (see also \cite{Retakh_Rubtsov_2010}).

\subsubsection{\texorpdfstring{$\PI$}{P1} case.} 

There is only one matrix $\PI$ system (see section \ref{P1}).
Being Hamiltonian, it is invariant with respect to the transposition $\tau$.

\subsection{Isomonodromic representations}
\label{sec:isompairs}

All scalar \Painleve systems admit an isomonodromic representation \cite{garnier1912equations, jimbo1981monodromy} of the form \eqref{eq:zerocurvcond}. Note that for commuting variables $u$ and $v$ the relation \eqref{eq:zerocurvcond} admits shifts of the form
\begin{equation}\label{BBB} 
    \mathbf{B} 
    \mapsto \mathbf{B}
    + q(z,u,v)\, \mathbf{I},
\end{equation}
where $q$ is an arbitrary polynomial in $u$ and $v$, while in the non-commutative case such transformations are not allowed.  

Using the non-abelianization procedure  proposed in \cite{Bobrova_Sokolov_2022},  we generalize well-known isomonodromic pairs for the scalar $\PVI - \PI$ systems to the non-abelian case. As a result, we obtain Lax pairs not only for specific \PPainleve-type systems with auxiliary autonomous systems listed in Sections \ref{P6} and \ref{P5-P1}, but for whole classes of systems satisfying the property 3 in Assumption \ref{assumpt1} (see Proposition  \ref{thm:intP6_nc} and item a) in Propositions \ref{thm:intP5_nc} -- \ref{thm:intP3_nc}).

The non-abelianization procedure is based on the observation that most of Lax pairs for scalar \PPainleve\,type systems are homogeneous in the following sense. If we represent $\mathbf{A}$ and $\mathbf{B}$ as rational functions in $\lambda$, $z$, and $\kappa_i$ with coefficients being $2\times 2$ matrices polynomial in $u$ and $v$, then these coefficients turn out to be homogeneous for proper weights of $u$ and $v$ (see the next section for details).
 
\subsubsection{Systems of \texorpdfstring{$\PVI$}{P6} type}\label{Sec5.2.1}

The Lax representation \eqref{eq:zerocurvcond} for the commutative $\PVI$ system  \eqref{eq:scalP6sys} found in~\cite{jimbo1981monodromy} is defined by   
\begin{align} \label{eq:matABform}
    \mathbf{A} (z, \lambda)
    &= \dfrac{A_0 (z)}{\lambda}
    + \dfrac{A_1 (z)}{\lambda - 1}
    + \dfrac{A_2 (z)}{\lambda - z},
    &
    \mathbf{B} (z, \lambda)
    &= B_{0} (z)
    - \dfrac{A_2 (z)}{\lambda - z},
\end{align}
where  the matrices $A_0 (z)$, $A_1 (z)$, $A_2 (z)$, and $B_{0} (z)$ are equivalent to
\begin{gather} 
    \notag
    \begin{aligned}
    A_0
    &= 
    \begin{pmatrix}
    - 1
    - \kappa_1
    + \kappa_4
    & 
    u z^{-1} - 1
    \\[0.9mm]
    0 & 0
    \end{pmatrix},
    &&&
    A_1
    &= 
    \begin{pmatrix}
    - u v + \kappa_1
    & 
    1 
    \\[0.9mm]
    - u^2 v^2 
    + \kappa_1 u v 
    + \kappa_3
    & 
    u v
    \end{pmatrix},
    \end{aligned}
    \\[2mm]
    \label{eq:scalLaxpair}
    A_2
    = 
    \begin{pmatrix}
    u v 
    + (\kappa_1 - \kappa_2 - \kappa_4) 
    & 
    - u z^{-1}
    \\[0.9mm]
    z u v^2 
    + (\kappa_1 - \kappa_2 - \kappa_4)  z v 
    & 
    - u v
    \end{pmatrix},
    \\[2mm]
    \notag
    B_{0}
    = 
    \begin{pmatrix}
    (z (z - 1))^{-1} \brackets{
    2 u^2 v - \kappa_1 u 
    - z \brackets{
    2 u v
    + (\kappa_1 - \kappa_2 - \kappa_4)
    }
    }
    & 0
    \\[0.9mm]
    - u v^2 - (\kappa_1 - \kappa_2 - \kappa_4)  v
    & 0
    \end{pmatrix}.
\end{gather}

If we assign the weights $w(u)=w(v)=1$ to the variables, then any matrix coefficient $Q(u,v)$ has homogeneous entries $q_{ij}$ such that $w(q_{11})=w(q_{22})=w(q_{12})+2=w(q_{21})-2.$

Under the non-abelianization we preserve the same structure of $\mathbf{A}$ and $\mathbf{B}$ in $\lambda$, $z$, and $\kappa_i$ and the homogeneity property. Besides, we require the in the scalar limit the pair coincides with the described above up to a transformation of the form \eqref{BBB}.

\begin{proposition}
\label{thm:ncLax_genP6}
Any non-abelian system of the $\PVI$ type described in Proposition {\rm\ref{thm:intP6_nc}}, is equivalent to the zero-curvature condition \eqref{eq:zerocurvcond} with matrices $\mathbf{A} (z, \lambda)$ and $\mathbf{B} (z, \lambda)$ of the form \eqref{eq:matABform}, where
    \begin{gather} 
    \notag
    \begin{aligned}
    A_0
    &= 
    \begin{pmatrix}
    - 1
    - \kappa_1
    + \kappa_4
    & 
    u z^{-1} - 1
    \\[0.9mm]
    0 & 0
    \end{pmatrix},
    &&&
    A_1
    &= 
    \begin{pmatrix}
    - v u + \kappa_1
    & 
    1 
    \\[0.9mm]
    - v u v u 
    + \kappa_1 v u 
    + \kappa_3
    & 
    v u
    \end{pmatrix},
    \end{aligned}
    \\[2mm]
    \label{eq:Laxpair_P6_nc}
    A_2
    = 
    \begin{pmatrix}
    v u 
    + (\kappa_1 - \kappa_2 - \kappa_4) 
    & 
    - u z^{-1}
    \\[0.9mm]
    z  \, v^2 u 
    + (\kappa_1 - \kappa_2 - \kappa_4) z v 
    & 
    - v u
    \end{pmatrix},
    \\[2mm]
    \notag
    \begin{aligned}
    B_{0}
    = 
    \begin{pmatrix}
    (z (z - 1))^{-1} \brackets{
    u v u + v u^2 
    - \kappa_1 u 
    - z \brackets{
    2 v u
    + (\kappa_1 - \kappa_2 - \kappa_4)
    }
    }
    & 0
    \\[0.9mm]
    - v^2 u - (\kappa_1 - \kappa_2 - \kappa_4)  v
    & 0
    \end{pmatrix}
    + q(z, u, v) \, \mathbf{I}
    ,
    \end{aligned}
    \\[1mm]
    \begin{aligned}
    z (z - 1) \, q (z, u, v)
    = 
    - a_1 u^2 v
    - (a_1 + a_2) u v u
    + (1 - a_1 - a_2 - a_3) v u^2
    - c_1 u v 
    + (1 + c_2) v u
    \\
    + \, e_1 u
    + z \brackets{- f_1 u v + (2 + f_2) v u - h_1 v}
    .
    \end{aligned}
\end{gather}
\vspace{-3mm}
\qed
\end{proposition}
In the scalar case the polynomial $q$ can be removed by shift \eqref{BBB}.

\subsubsection{Systems of \texorpdfstring{$\PV$}{P5} type}\label{Sec5.2.2}

The \PPainleve-5 system associated with the Hamiltonian \eqref{eq:ham5} admits an isomonodromic Lax pair of the form
\begin{align} 
    \label{eq:matABform_P5}
    \mathbf{A} (\lambda, z)
    &= A_0 (z)
    + \frac{A_1 (z)}{\lambda}
    + \frac{A_2 (z)}{\lambda - 1},
    &
    \mathbf{B} (\lambda, z)
    &= B_1 \lambda
    + B_{0} (z),
\end{align}
where \cite{jimbo1981monodromy}
\begin{gather} 
    \notag
    \begin{aligned}
    A_{0}
    &= 
    \begin{pmatrix}
    \kappa_4 z & 0 \\[0.9mm] 0 & 0
    \end{pmatrix},
    &
    A_1
    &= 
    \begin{pmatrix}
    - u v + \kappa_1 & 1
    \\[0.9mm]
    - u^2 v^2
    + \kappa_1 u v
    + \kappa_3
    & u v
    \end{pmatrix},
    &
    A_2
    &= 
    \begin{pmatrix}
    u v - \kappa_2 & - u
    \\[0.9mm]
    u v^2 - \kappa_2 v & - u v
    \end{pmatrix},
    \end{aligned}
    \\[-1mm]
    \label{eq:Laxpair_P50}
    \\[-1mm]
    \notag
    \begin{aligned}
    B_1
    &= 
    \begin{pmatrix}
    \kappa_4 & 0 \\[0.9mm] 0 & 0
    \end{pmatrix},
    &
    B_{0}
    &= z^{-1}
    \begin{pmatrix}
    2 u^2 v
    - 2 u v
    - \kappa_1 u 
    + \kappa_1
    & 
    - u + 1
    \\[0.9mm]
    - u^2 v^2
    + u v^2
    + \kappa_1 u v 
    - \kappa_2 v
    + \kappa_3
    &
    0
    \end{pmatrix}.
    \end{aligned}
\end{gather}
This Lax pair can be generalized to the non-abelian case as follows 
\begin{proposition}
\label{thm:ncLax_genP5}
Let a non-abelian system of the $\PV$ type possesses an auxiliary autonomous system described in the item {\rm{a)}} of Proposition {\rm\ref{thm:intP5_nc}}. 
Then it has an isomonodromic Lax pair of the form \eqref{eq:matABform_P5} with matrices
    \begin{gather} 
    \notag
    \begin{aligned}
    A_0
    &= 
    \begin{pmatrix}
    \kappa_4 z
    & 
    0
    \\[0.9mm]
    0 & 0
    \end{pmatrix},
    &&&
    A_1
    &= 
    \begin{pmatrix}
    - v u + \kappa_1
    & 
    1 
    \\[0.9mm]
    - v u v u 
    + \kappa_1 v u 
    + \kappa_3
    & 
    v u
    \end{pmatrix},
    &&&
    A_2
    =& 
    \begin{pmatrix}
    v u 
    - \kappa_2
    & 
    - u
    \\[0.9mm]
    v^2 u 
    - \kappa_2 v 
    & 
    - v u
    \end{pmatrix},
    \end{aligned}
    \\[2mm]
    \begin{aligned}
    B_1
    &= 
    \begin{pmatrix}
    \kappa_4
    & 
    0
    \\[0.9mm]
    0 & 0
    \end{pmatrix},
    &&&
    B_{0}
    &= z^{-1}
    \begin{pmatrix}
    u v u + v u^2 
    - 2 v u
    - \kappa_1 u + \kappa_1
    & - u + 1
    \\[0.9mm]
    - v u v u
    + v^2 u
    + \kappa_1 v u 
    - \kappa_2 v 
    + \kappa_3
    & 0
    \end{pmatrix}
    + q(z, u, v) \, \mathbf{I}
    ,
    \end{aligned}
    \\
    \begin{aligned}
    z \, q (z, u, v)
    = 
    - a_1 u^2 v
    - (a_1 + a_2) u v u
    + (1 - a_1 - a_2 - a_3) v u^2
    - c_1 u v 
    + (3 + c_2) v u
    \\
    + \, f_1 u
    - e_1 v
    .
    \end{aligned}
\end{gather}
\vspace{-3mm}
\qed
\end{proposition}

\subsubsection{Systems of \texorpdfstring{$\PIV$}{P4} type}\label{P4}

The Hamiltonian \eqref{eq:ham4} leads to the \PPainleve-4 system. Its well-known \cite{jimbo1981monodromy} that this system is equivalent to the zero-curvature condition \eqref{eq:zerocurvcond} with
\begin{align} 
\label{eq:matABform_P4}
    \mathbf{A} (\lambda, z)
    &= A_1 \lambda 
    + A_{0} (z)
    + A_{-1} (z) \lambda^{-1},
    &
    \mathbf{B} (\lambda, z)
    &= B_1 \lambda 
    + B_{0} (z)
    ,
\end{align}
where the matrices $A_1$, $A_0 (z)$, $A_{-1} (z)$, $B_1$, and $B_0 (z)$ are given by
\begin{equation} 
\label{eq:Laxpair_P4}
    \begin{gathered}
        \begin{aligned}
        A_1
        &= 
        \begin{pmatrix}
        - 2 & 0 \\[0.9mm] 0 & 0
        \end{pmatrix},
        &&&
        A_0
        &= 
        \begin{pmatrix}
        - 2 z & 1 \\[0.9mm] 
        u v + \kappa_3 & 0
        \end{pmatrix},
        &&&
        A_{-1}    
        &= \tfrac12
        \begin{pmatrix}
        u v + \kappa_2 & - u \\[0.9mm]
        u v^2 + \kappa_2 v & - u v
        \end{pmatrix},
        \end{aligned}
        \\[2mm]
        \begin{aligned}
        B_1
        &= 
        \begin{pmatrix}
        - 2 & 0 \\[0.9mm] 0 & 0
        \end{pmatrix},
        &&&
        B_0
        &= 
        \begin{pmatrix}
        - u - 2 z & 1 \\[0.9mm]
        u v + \kappa_3 & 0
        \end{pmatrix}.
        \end{aligned}
    \end{gathered}
\end{equation}
According to the item a) from Proposition \ref{thm:sysP4_nc}, non-abelian systems of the $\PIV$ type associated with \eqref{eq:sysP4_nc} admit the non-abelianization of the first integrals. One can verify that they also have an isomonodromic representation.
\begin{proposition}
\label{thm:ncLax_genP4}
For any $\alpha$, $\beta \in \mathbb{C}$, a non-abelian $\PIV$ type system defined by the autonomous system \eqref{eq:sysP4_nc} has a Lax pair of the form \eqref{eq:matABform_P4}, where 
    \begin{gather} 
    \notag
    \begin{aligned}
        A_1
        &= 
        \begin{pmatrix}
        - 2 & 0 \\[0.9mm] 0 & 0
        \end{pmatrix},
        &&&
        A_0
        &= 
        \begin{pmatrix}
        - 2 z & 1 \\[0.9mm] 
        v u + \kappa_3 & 0
        \end{pmatrix},
        &&&
        A_{-1}    
        &= \tfrac12
        \begin{pmatrix}
        v u + \kappa_2 & - u \\[0.9mm]
        v^2 u + \kappa_2 v & - v u
        \end{pmatrix},
    \end{aligned}
    \\[2mm]
    \begin{aligned}
    B_1
    &= 
    \begin{pmatrix}
        - 2 & 0 \\[0.9mm] 0 & 0
    \end{pmatrix},
    &&&
    B_{0}
    &=
    \begin{pmatrix}
    - u - 2 z & 1 
    \\[0.9mm]
    v u + \kappa_3 & 0
    \end{pmatrix}
    + q(z, u, v) \, \mathbf{I}
    ,
    &&&
    q (z, u, v)
    &= - \beta u - (\alpha + 2) v
    .
    \end{aligned}
\end{gather}
\vspace{-3mm}
\qed
\end{proposition}

\subsubsection{\texorpdfstring{$\PIII$}{P3} type systems} 
The scalar $\PIIIpr$ system related to Hamiltonian \eqref{eq:ham3} has (see \cite{jimbo1981monodromy}) the following isomonodromic Lax pair
\begin{align} \label{eq:matABform_P3D6'_}
    \mathbf{A} (\lambda, z)
    &= A_0 (z)
    + A_{-1} (z) \lambda^{-1} 
    + A_{-2} (z) \lambda^{-2},
    &
    \mathbf{B} (\lambda, z)
    &= B_1 \lambda + B_0 (z),
\end{align}
where the matrices $A_0 (z)$, $A_{-1} (z)$, $A_{-2} (z)$, $B_1$, and $B_0 (z)$ are equivalent to
\begin{gather}
    \notag
    \begin{aligned}
    A_0
    &= 
    \begin{pmatrix}
    \kappa_4 z & 0 \\[0.9mm]
    0 & 0
    \end{pmatrix},
    &
    A_{-1}
    &= -
    \begin{pmatrix}
    \kappa_1
    & 
    u
    \\[0.9mm]
    u v^2
    + \kappa_1 v
    + \kappa_2 u v 
    + \kappa_3
    & 
    0
    \end{pmatrix},
    &
    A_{-2}
    &= 
    \begin{pmatrix}
    v + \kappa_2
    & 
    - 1
    \\[0.9mm]
    v^2 + \kappa_2 v
    &
    - v
    \end{pmatrix},
    \end{aligned}
    \\[-1mm]
    \label{eq:Laxpair_P3D6}
    \\[-1mm]
    \notag
    \begin{aligned}
    B_1
    &= 
    \begin{pmatrix}
    \kappa_4 & 0 
    \\[0.9mm]
    0 & 0
    \end{pmatrix}
    ,
    &&&
    B_0
    &= z^{-1}
    \begin{pmatrix}
    2 u v 
    + \kappa_2 u 
    &
    - u
    \\[0.9mm]
    - \brackets{
    u v^2 
    + \kappa_1 v
    + \kappa_2 u v 
    + \kappa_3
    }
    & 
    0
    \end{pmatrix}
    .
    \end{aligned}
\end{gather}
\begin{proposition}
\label{thm:ncLax_genP3}
Suppose that a non-abelian system of the $\PIIIpr$ type has an auxiliary autonomous system given by the item {\rm{a)}} of Proposition {\rm\ref{thm:intP3_nc}}. Then it is equivalent to \eqref{eq:zerocurvcond}, where a Lax pair has the form \eqref{eq:matABform_P3D6'_} with
    \begin{gather} 
    \begin{aligned}
    A_0
    &= 
    \begin{pmatrix}
    \kappa_4 z & 0 \\[0.9mm]
    0 & 0
    \end{pmatrix},
    &
    A_{-1}
    &= -
    \begin{pmatrix}
    \kappa_1
    & 
    u
    \\[0.9mm]
    v u v
    + \kappa_1 v
    + \kappa_2 v u 
    + \kappa_3
    & 
    0
    \end{pmatrix},
    &
    A_{-2}
    &= 
    \begin{pmatrix}
    v + \kappa_2
    & 
    - 1
    \\[0.9mm]
    v^2 + \kappa_2 v
    &
    - v
    \end{pmatrix},
    \end{aligned}
    \\[2mm]
    \notag
    \begin{aligned}
    B_1
    &= 
    \begin{pmatrix}
    \kappa_4 & 0 
    \\[0.9mm]
    0 & 0
    \end{pmatrix},
    &&&
    B_{0}
    &= z^{-1}
    \begin{pmatrix}
    u v + v u
    + \kappa_2 u 
    &
    - u
    \\[0.9mm]
    - \brackets{
    v u v
    + \kappa_1 v
    + \kappa_2 v u
    + \kappa_3
    }
    & 
    0
    \end{pmatrix}
    + q(z, u, v) \, \mathbf{I}
    ,
    \end{aligned}
    \\[2mm]
    \begin{aligned}
    z \, q (z, u, v)
    = - a_1 u v 
    + (a_2 - 1) v u
    + c_1 u
    .
    \end{aligned}
\end{gather}
\vspace{-3mm}
\qed
\end{proposition}
Special cases of $\PIIIpr$ type systems are considered in Appendix \ref{P3'D7}.

\subsubsection{Systems of \texorpdfstring{$\PII$}{P2} type}
The commutative \PPainleve-2 system \eqref{eq:sysP2} possesses an isomonodromic Lax pair of the form~\cite{jimbo1981monodromy} 
\begin{align} \label{eq:matABform_P2}
    \mathbf{A} (\lambda, z)
    &= A_2 \lambda^2 
    + A_1 (z) \lambda 
    + A_0 (z),
    &
    \mathbf{B} (\lambda, z)
    &= B_1 \lambda + B_0 (z),
\end{align}
where the matrices $A_2$, $A_1 (z)$, $A_0(z)$, $B_1$, and $B_0 (z)$ are equivalent to the following:
\begin{equation} 
\label{eq:Laxpair_P2}
    \begin{gathered}
        \begin{aligned}
        A_2
        &= 
        \begin{pmatrix}
        2 & 0 \\[0.9mm] 0 & 0
        \end{pmatrix},
        &&&
        A_1
        &= 
        \begin{pmatrix}
        0 & -2 \\[0.9mm]
        - v & 0
        \end{pmatrix},
        &&&
        A_0
        &= 
        \begin{pmatrix}
        - v + z & - 2 u
        \\[0.9mm]
        u v + \kappa & v
        \end{pmatrix},
        \end{aligned}
        \\[2mm]
        \begin{aligned}
        B_1
        &= 
        \begin{pmatrix}
        1 & 0 \\[0.9mm] 0 & 0
        \end{pmatrix},
        &&&
        B_0
        &= 
        \begin{pmatrix}
        - u & -1 \\[0.9mm]
        - \tfrac12 v & 0
        \end{pmatrix}.
        \end{aligned}
    \end{gathered}
\end{equation}
\begin{proposition}
\label{thm:ncLax_genP2}
Any non-abelian system of the $\PII$ corresponding to autonomous system \eqref{eq:sysP2_nc}, is equivalent to the zero-curvature condition \eqref{eq:zerocurvcond} with matrices $\mathbf{A} (z, \lambda)$ and $\mathbf{B} (z, \lambda)$ of the form \eqref{eq:matABform_P2}, where
    \begin{gather} 
    \notag
    \begin{aligned}
        A_2
        &= 
        \begin{pmatrix}
        2 & 0 \\[0.9mm] 0 & 0
        \end{pmatrix},
        &&&
        A_1
        &= 
        \begin{pmatrix}
        0 & -2 \\[0.9mm]
        - v & 0
        \end{pmatrix},
        &&&
        A_0
        &= 
        \begin{pmatrix}
        - v + z & - 2 u
        \\[0.9mm]
        v u + \kappa & v
        \end{pmatrix},
    \end{aligned}
    \\[2mm]
    \notag
    \begin{aligned}
    B_1
    &= 
    \begin{pmatrix}
    1 & 0 \\[0.9mm] 0 & 0
    \end{pmatrix},
    &&&
    B_{0}
    &= 
    \begin{pmatrix}
    - u
    & - 1
    \\[0.9mm]
    - \tfrac12 v
    & 0
    \end{pmatrix}
    + q(z, u, v) \, \mathbf{I}
    ,
    &&&
    q (z, u, v)
    &= 
    (1 - \beta) u
    .
    \end{aligned}
\end{gather}
\qed
\end{proposition}

\begin{remark}
Here and in Subsection \ref{sec:P2_case}, we considered the Lax pairs of  Jimbo-Miwa  {\rm\cite{jimbo1981monodromy}} and the Flaschka-Newell types for the \PPainleve-2 equation.   It was shown in  {\rm\cite{suleimanov2008quantizations, joshi2009linearization}} that in the scalar case these Lax pairs are related to each other by a  generalized Laplace transform supplemented with a gauge transformation. Formally, assuming the cancellation of terms that arise from the integration by parts, it is easy to verify that this correspondence can be generalized to the non-abelian case.
\end{remark}

\subsubsection{\texorpdfstring{$\PI$}{P1} type system}\label{Sec5.2.6}
The $\PI$ system has the  Lax pair given by
\begin{align} \label{eq:matABform_P1}
    \mathbf{A} (\lambda, z)
    &= A_2 \lambda^2 
    + A_1 (z) \lambda 
    + A_0 (z),
    &
    \mathbf{B} (\lambda, z)
    &= B_1 \lambda + B_0 (z)
\end{align}
with
\begin{equation} 
\label{eq:Laxpair_P1}
    \begin{gathered}
        \begin{aligned}
        A_2
        &= 
        \begin{pmatrix}
        0 & 0 \\[0.9mm] 2 & 0
        \end{pmatrix},
        &&&
        A_1
        &= 
        \begin{pmatrix}
        0 & -2 \\[0.9mm]
        - 2 u & 0
        \end{pmatrix},
        &&&
        A_0
        &= 
        \begin{pmatrix}
        - v & - 2 u
        \\[0.9mm]
        2 u^2 + z & v
        \end{pmatrix},
        \end{aligned}
        \\[2mm]
        \begin{aligned}
        B_1
        &= 
        \begin{pmatrix}
        0 & 0 \\[0.9mm] 1 & 0
        \end{pmatrix},
        &&&
        B_0
        &= 
        \begin{pmatrix}
        0 & -1 \\[0.9mm]
        - 2 u & 0
        \end{pmatrix}
        .
        \end{aligned}
    \end{gathered}
\end{equation}
It can be generalized to the non-abelian case by replacing the variables $u$ and $v$ with non-commutative ones.  

\subsection{Tree of degenerations}
\label{sec:deg}
In \cite{bobrova2022non} limiting transitions for non-abelian \PPainleve-type systems having Okamoto integrals were found. The same formulas with minor changes in constants 
can be applied for systems from the current paper\footnote{In particular,  one should replace  $\kappa_3$ with $\kappa$ for the $\PII$ systems.}.
 
The degeneration trees for $\PVI$ systems \ref{eq:P6_1} -- \ref{eq:P6_16} from Appendix \ref{sec:sysintlistP6} non-equivalent with respect to the transposition $\tau$  are shown schematically in Figures \ref{pic:deg_orbit13}, \ref{pic:deg_orbit2}, and \ref{pic:deg_orbit4}. 
\begin{figure}[H]
    \centering
    \begin{minipage}[l]{0.49\linewidth}
    \begin{figure}[H]
        \hspace{-0.8cm}
        \scalebox{0.93}{\tikzset{every picture/.style={line width=0.75pt}} 

\begin{tikzpicture}[x=0.75pt,y=0.75pt,yscale=-1,xscale=1]

\draw    (218.01,129) -- (245.25,129) ;
\draw [shift={(247.25,129)}, rotate = 180] [color={rgb, 255:red, 0; green, 0; blue, 0 }  ][line width=0.75]    (10.93,-3.29) .. controls (6.95,-1.4) and (3.31,-0.3) .. (0,0) .. controls (3.31,0.3) and (6.95,1.4) .. (10.93,3.29)   ;
\draw    (288.01,116) -- (315.25,116) ;
\draw [shift={(317.25,116)}, rotate = 180] [color={rgb, 255:red, 0; green, 0; blue, 0 }  ][line width=0.75]    (10.93,-3.29) .. controls (6.95,-1.4) and (3.31,-0.3) .. (0,0) .. controls (3.31,0.3) and (6.95,1.4) .. (10.93,3.29)   ;
\draw    (288.01,141) -- (315.25,141) ;
\draw [shift={(317.25,141)}, rotate = 180] [color={rgb, 255:red, 0; green, 0; blue, 0 }  ][line width=0.75]    (10.93,-3.29) .. controls (6.95,-1.4) and (3.31,-0.3) .. (0,0) .. controls (3.31,0.3) and (6.95,1.4) .. (10.93,3.29)   ;
\draw    (288.01,102) -- (316.19,89.23) ;
\draw [shift={(318.01,88.4)}, rotate = 155.61] [color={rgb, 255:red, 0; green, 0; blue, 0 }  ][line width=0.75]    (10.93,-3.29) .. controls (6.95,-1.4) and (3.31,-0.3) .. (0,0) .. controls (3.31,0.3) and (6.95,1.4) .. (10.93,3.29)   ;
\draw    (288.01,156.96) -- (316.19,173.89) ;
\draw [shift={(317.9,174.93)}, rotate = 211.01] [color={rgb, 255:red, 0; green, 0; blue, 0 }  ][line width=0.75]    (10.93,-3.29) .. controls (6.95,-1.4) and (3.31,-0.3) .. (0,0) .. controls (3.31,0.3) and (6.95,1.4) .. (10.93,3.29)   ;
\draw    (359.01,177) -- (386.25,177) ;
\draw [shift={(388.25,177)}, rotate = 180] [color={rgb, 255:red, 0; green, 0; blue, 0 }  ][line width=0.75]    (10.93,-3.29) .. controls (6.95,-1.4) and (3.31,-0.3) .. (0,0) .. controls (3.31,0.3) and (6.95,1.4) .. (10.93,3.29)   ;
\draw    (359.01,80) -- (386.25,80) ;
\draw [shift={(388.25,80)}, rotate = 180] [color={rgb, 255:red, 0; green, 0; blue, 0 }  ][line width=0.75]    (10.93,-3.29) .. controls (6.95,-1.4) and (3.31,-0.3) .. (0,0) .. controls (3.31,0.3) and (6.95,1.4) .. (10.93,3.29)   ;
\draw    (359.01,112) -- (387.19,99.23) ;
\draw [shift={(389.01,98.4)}, rotate = 155.61] [color={rgb, 255:red, 0; green, 0; blue, 0 }  ][line width=0.75]    (10.93,-3.29) .. controls (6.95,-1.4) and (3.31,-0.3) .. (0,0) .. controls (3.31,0.3) and (6.95,1.4) .. (10.93,3.29)   ;
\draw    (359.01,143.96) -- (387.19,160.89) ;
\draw [shift={(388.9,161.93)}, rotate = 211.01] [color={rgb, 255:red, 0; green, 0; blue, 0 }  ][line width=0.75]    (10.93,-3.29) .. controls (6.95,-1.4) and (3.31,-0.3) .. (0,0) .. controls (3.31,0.3) and (6.95,1.4) .. (10.93,3.29)   ;
\draw    (429.01,86) -- (457.19,102.93) ;
\draw [shift={(458.9,103.97)}, rotate = 211.01] [color={rgb, 255:red, 0; green, 0; blue, 0 }  ][line width=0.75]    (10.93,-3.29) .. controls (6.95,-1.4) and (3.31,-0.3) .. (0,0) .. controls (3.31,0.3) and (6.95,1.4) .. (10.93,3.29)   ;
\draw    (429.01,174) -- (457.19,161.23) ;
\draw [shift={(459.01,160.4)}, rotate = 155.61] [color={rgb, 255:red, 0; green, 0; blue, 0 }  ][line width=0.75]    (10.93,-3.29) .. controls (6.95,-1.4) and (3.31,-0.3) .. (0,0) .. controls (3.31,0.3) and (6.95,1.4) .. (10.93,3.29)   ;

\draw (325.01,71.37) node [anchor=north west][inner sep=0.75pt]    {$\PPIVn{4}$};
\draw (325.01,105.4) node [anchor=north west][inner sep=0.75pt]    {$\PPIIIprn{1}$};
\draw (325.01,133.4) node [anchor=north west][inner sep=0.75pt]    {$\PPIVn{3}$};
\draw (325.01,167.4) node [anchor=north west][inner sep=0.75pt]    {$\PPIIIprn{2}$};
\draw (255.01,119.4) node [anchor=north west][inner sep=0.75pt]    {\ref{eq:P5_1}};
\draw (185.01,119.4) node [anchor=north west][inner sep=0.75pt]    {\ref{eq:P6_1}};
\draw (395.01,71.39) node [anchor=north west][inner sep=0.75pt]    {$\PPIIn{1}$};
\draw (395.01,167.4) node [anchor=north west][inner sep=0.75pt]    {$\PPIIn{2}$};
\draw (465.01,119.4) node [anchor=north west][inner sep=0.75pt]    {$\PIn{H}$};

\end{tikzpicture}}
    \end{figure}
    \end{minipage}
    \begin{minipage}[l]{0.49\linewidth}
    \begin{figure}[H]
        \hspace{0.3cm}
        \scalebox{0.93}{\tikzset{every picture/.style={line width=0.75pt}} 

\begin{tikzpicture}[x=0.75pt,y=0.75pt,yscale=-1,xscale=1]

\draw    (218.01,129) -- (245.25,129) ;
\draw [shift={(247.25,129)}, rotate = 180] [color={rgb, 255:red, 0; green, 0; blue, 0 }  ][line width=0.75]    (10.93,-3.29) .. controls (6.95,-1.4) and (3.31,-0.3) .. (0,0) .. controls (3.31,0.3) and (6.95,1.4) .. (10.93,3.29)   ;
\draw    (288.01,116) -- (315.25,116) ;
\draw [shift={(317.25,116)}, rotate = 180] [color={rgb, 255:red, 0; green, 0; blue, 0 }  ][line width=0.75]    (10.93,-3.29) .. controls (6.95,-1.4) and (3.31,-0.3) .. (0,0) .. controls (3.31,0.3) and (6.95,1.4) .. (10.93,3.29)   ;
\draw    (288.01,141) -- (315.25,141) ;
\draw [shift={(317.25,141)}, rotate = 180] [color={rgb, 255:red, 0; green, 0; blue, 0 }  ][line width=0.75]    (10.93,-3.29) .. controls (6.95,-1.4) and (3.31,-0.3) .. (0,0) .. controls (3.31,0.3) and (6.95,1.4) .. (10.93,3.29)   ;
\draw    (288.01,102) -- (316.19,89.23) ;
\draw [shift={(318.01,88.4)}, rotate = 155.61] [color={rgb, 255:red, 0; green, 0; blue, 0 }  ][line width=0.75]    (10.93,-3.29) .. controls (6.95,-1.4) and (3.31,-0.3) .. (0,0) .. controls (3.31,0.3) and (6.95,1.4) .. (10.93,3.29)   ;
\draw    (288.01,156.96) -- (316.19,173.89) ;
\draw [shift={(317.9,174.93)}, rotate = 211.01] [color={rgb, 255:red, 0; green, 0; blue, 0 }  ][line width=0.75]    (10.93,-3.29) .. controls (6.95,-1.4) and (3.31,-0.3) .. (0,0) .. controls (3.31,0.3) and (6.95,1.4) .. (10.93,3.29)   ;
\draw    (359.01,177) -- (386.25,177) ;
\draw [shift={(388.25,177)}, rotate = 180] [color={rgb, 255:red, 0; green, 0; blue, 0 }  ][line width=0.75]    (10.93,-3.29) .. controls (6.95,-1.4) and (3.31,-0.3) .. (0,0) .. controls (3.31,0.3) and (6.95,1.4) .. (10.93,3.29)   ;
\draw    (359.01,80) -- (386.25,80) ;
\draw [shift={(388.25,80)}, rotate = 180] [color={rgb, 255:red, 0; green, 0; blue, 0 }  ][line width=0.75]    (10.93,-3.29) .. controls (6.95,-1.4) and (3.31,-0.3) .. (0,0) .. controls (3.31,0.3) and (6.95,1.4) .. (10.93,3.29)   ;
\draw    (359.01,112) -- (387.19,99.23) ;
\draw [shift={(389.01,98.4)}, rotate = 155.61] [color={rgb, 255:red, 0; green, 0; blue, 0 }  ][line width=0.75]    (10.93,-3.29) .. controls (6.95,-1.4) and (3.31,-0.3) .. (0,0) .. controls (3.31,0.3) and (6.95,1.4) .. (10.93,3.29)   ;
\draw    (359.01,143.96) -- (387.19,160.89) ;
\draw [shift={(388.9,161.93)}, rotate = 211.01] [color={rgb, 255:red, 0; green, 0; blue, 0 }  ][line width=0.75]    (10.93,-3.29) .. controls (6.95,-1.4) and (3.31,-0.3) .. (0,0) .. controls (3.31,0.3) and (6.95,1.4) .. (10.93,3.29)   ;
\draw    (429.01,86) -- (457.19,102.93) ;
\draw [shift={(458.9,103.97)}, rotate = 211.01] [color={rgb, 255:red, 0; green, 0; blue, 0 }  ][line width=0.75]    (10.93,-3.29) .. controls (6.95,-1.4) and (3.31,-0.3) .. (0,0) .. controls (3.31,0.3) and (6.95,1.4) .. (10.93,3.29)   ;
\draw    (429.01,174) -- (457.19,161.23) ;
\draw [shift={(459.01,160.4)}, rotate = 155.61] [color={rgb, 255:red, 0; green, 0; blue, 0 }  ][line width=0.75]    (10.93,-3.29) .. controls (6.95,-1.4) and (3.31,-0.3) .. (0,0) .. controls (3.31,0.3) and (6.95,1.4) .. (10.93,3.29)   ;

\draw (325.01,71.37) node [anchor=north west][inner sep=0.75pt]    {\ref{eq:P4_0_m3}};
\draw (325.01,105.4) node [anchor=north west][inner sep=0.75pt]    {\ref{eq:P3_4}};
\draw (325.01,133.4) node [anchor=north west][inner sep=0.75pt]    {$\PPIVn{5}$};
\draw (325.01,167.4) node [anchor=north west][inner sep=0.75pt]    {$\PPIIIprn{8}$};
\draw (255.01,119.4) node [anchor=north west][inner sep=0.75pt]    {\ref{eq:P5_5}};
\draw (185.01,119.4) node [anchor=north west][inner sep=0.75pt]    {\ref{eq:P6_8}};
\draw (395.01,71.39) node [anchor=north west][inner sep=0.75pt]    {\ref{eq:P2_m3}};
\draw (395.01,167.4) node [anchor=north west][inner sep=0.75pt]    {$\PPIIn{1}$};
\draw (465.01,119.4) node [anchor=north west][inner sep=0.75pt]    {$\PIn{H}$};

\end{tikzpicture}}
    \end{figure}
    \end{minipage}
    \caption{Degenerations of \ref{eq:P6_1} and \ref{eq:P6_8}}
    \label{pic:deg_orbit13}
\end{figure}

\begin{figure}[H]
    \centering
    \scalebox{0.93}{\tikzset{every picture/.style={line width=0.75pt}} 

\begin{tikzpicture}[x=0.75pt,y=0.75pt,yscale=-1,xscale=1]

\draw    (145,241) -- (172.23,241) ;
\draw [shift={(174.23,241)}, rotate = 180] [color={rgb, 255:red, 0; green, 0; blue, 0 }  ][line width=0.75]    (10.93,-3.29) .. controls (6.95,-1.4) and (3.31,-0.3) .. (0,0) .. controls (3.31,0.3) and (6.95,1.4) .. (10.93,3.29)   ;
\draw    (145,225) -- (173.18,212.23) ;
\draw [shift={(175,211.4)}, rotate = 155.61] [color={rgb, 255:red, 0; green, 0; blue, 0 }  ][line width=0.75]    (10.93,-3.29) .. controls (6.95,-1.4) and (3.31,-0.3) .. (0,0) .. controls (3.31,0.3) and (6.95,1.4) .. (10.93,3.29)   ;
\draw    (145,150) -- (172.23,150) ;
\draw [shift={(174.23,150)}, rotate = 180] [color={rgb, 255:red, 0; green, 0; blue, 0 }  ][line width=0.75]    (10.93,-3.29) .. controls (6.95,-1.4) and (3.31,-0.3) .. (0,0) .. controls (3.31,0.3) and (6.95,1.4) .. (10.93,3.29)   ;
\draw    (145,161) -- (173.18,177.93) ;
\draw [shift={(174.89,178.97)}, rotate = 211.01] [color={rgb, 255:red, 0; green, 0; blue, 0 }  ][line width=0.75]    (10.93,-3.29) .. controls (6.95,-1.4) and (3.31,-0.3) .. (0,0) .. controls (3.31,0.3) and (6.95,1.4) .. (10.93,3.29)   ;
\draw    (74,241) -- (101.23,241) ;
\draw [shift={(103.23,241)}, rotate = 180] [color={rgb, 255:red, 0; green, 0; blue, 0 }  ][line width=0.75]    (10.93,-3.29) .. controls (6.95,-1.4) and (3.31,-0.3) .. (0,0) .. controls (3.31,0.3) and (6.95,1.4) .. (10.93,3.29)   ;
\draw    (74,150) -- (101.23,150) ;
\draw [shift={(103.23,150)}, rotate = 180] [color={rgb, 255:red, 0; green, 0; blue, 0 }  ][line width=0.75]    (10.93,-3.29) .. controls (6.95,-1.4) and (3.31,-0.3) .. (0,0) .. controls (3.31,0.3) and (6.95,1.4) .. (10.93,3.29)   ;
\draw    (74,179) -- (102.18,166.23) ;
\draw [shift={(104,165.4)}, rotate = 155.61] [color={rgb, 255:red, 0; green, 0; blue, 0 }  ][line width=0.75]    (10.93,-3.29) .. controls (6.95,-1.4) and (3.31,-0.3) .. (0,0) .. controls (3.31,0.3) and (6.95,1.4) .. (10.93,3.29)   ;
\draw    (225,236) -- (253.18,223.23) ;
\draw [shift={(255,222.4)}, rotate = 155.61] [color={rgb, 255:red, 0; green, 0; blue, 0 }  ][line width=0.75]    (10.93,-3.29) .. controls (6.95,-1.4) and (3.31,-0.3) .. (0,0) .. controls (3.31,0.3) and (6.95,1.4) .. (10.93,3.29)   ;
\draw [color={rgb, 255:red, 0; green, 0; blue, 0 }  ,draw opacity=1 ]   (297,191) -- (324.23,191) ;
\draw [shift={(326.23,191)}, rotate = 180] [color={rgb, 255:red, 0; green, 0; blue, 0 }  ,draw opacity=1 ][line width=0.75]    (10.93,-3.29) .. controls (6.95,-1.4) and (3.31,-0.3) .. (0,0) .. controls (3.31,0.3) and (6.95,1.4) .. (10.93,3.29)   ;
\draw    (225,150) -- (241.82,160.11) -- (253.18,166.93) ;
\draw [shift={(254.89,167.97)}, rotate = 211.01] [color={rgb, 255:red, 0; green, 0; blue, 0 }  ][line width=0.75]    (10.93,-3.29) .. controls (6.95,-1.4) and (3.31,-0.3) .. (0,0) .. controls (3.31,0.3) and (6.95,1.4) .. (10.93,3.29)   ;
\draw [color={rgb, 255:red, 0; green, 0; blue, 0 }  ,draw opacity=1 ]   (225,180) -- (252.23,180) ;
\draw [shift={(254.23,180)}, rotate = 180] [color={rgb, 255:red, 0; green, 0; blue, 0 }  ,draw opacity=1 ][line width=0.75]    (10.93,-3.29) .. controls (6.95,-1.4) and (3.31,-0.3) .. (0,0) .. controls (3.31,0.3) and (6.95,1.4) .. (10.93,3.29)   ;
\draw [color={rgb, 255:red, 0; green, 0; blue, 0 }  ,draw opacity=1 ]   (225,201) -- (252.23,201) ;
\draw [shift={(254.23,201)}, rotate = 180] [color={rgb, 255:red, 0; green, 0; blue, 0 }  ,draw opacity=1 ][line width=0.75]    (10.93,-3.29) .. controls (6.95,-1.4) and (3.31,-0.3) .. (0,0) .. controls (3.31,0.3) and (6.95,1.4) .. (10.93,3.29)   ;

\draw (187,138.4) node [anchor=north west][inner sep=0.75pt]    {$\PPIIIprn{8}$};
\draw (187,167.4) node [anchor=north west][inner sep=0.75pt]    {$\PPIVn{3}$};
\draw (187,197.4) node [anchor=north west][inner sep=0.75pt]    {$\PPIIIprn{1}$};
\draw (187,231.4) node [anchor=north west][inner sep=0.75pt]    {\ref{eq:P4_1_m2}};
\draw (111,139.4) node [anchor=north west][inner sep=0.75pt]    {$\PPVn{9}$};
\draw (111,231.4) node [anchor=north west][inner sep=0.75pt]    {\ref{eq:P5_2}};
\draw (39,139.4) node [anchor=north west][inner sep=0.75pt]    {\ref{eq:P6_7}};
\draw (40,182.4) node [anchor=north west][inner sep=0.75pt]    {\ref{eq:P6_4}};
\draw (40,231.4) node [anchor=north west][inner sep=0.75pt]    {\ref{eq:P6_2}};
\draw (264,182.4) node [anchor=north west][inner sep=0.75pt]    {$\PPIIn{1}$};
\draw (328,182.4) node [anchor=north west][inner sep=0.75pt]    {$\PIn{H}$};

\end{tikzpicture}}
    \caption{Degenerations of \ref{eq:P6_2}, \ref{eq:P6_4}, and \ref{eq:P6_7}}
    \label{pic:deg_orbit2}
\end{figure}
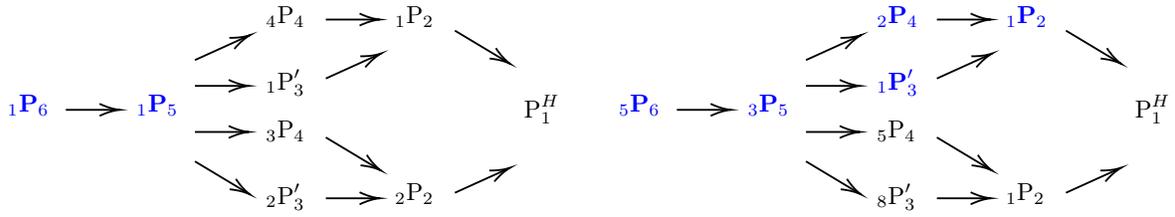

\begin{figure}[H]
    \centering
    \begin{minipage}[l]{0.49\linewidth}
    \begin{figure}[H]
        \hspace{-0.8cm}
        \scalebox{0.93}{\tikzset{every picture/.style={line width=0.75pt}} 

\begin{tikzpicture}[x=0.75pt,y=0.75pt,yscale=-1,xscale=1]

\draw    (217.01,116) -- (226.4,116) -- (244.25,116) ;
\draw [shift={(246.25,116)}, rotate = 180] [color={rgb, 255:red, 0; green, 0; blue, 0 }  ][line width=0.75]    (10.93,-3.29) .. controls (6.95,-1.4) and (3.31,-0.3) .. (0,0) .. controls (3.31,0.3) and (6.95,1.4) .. (10.93,3.29)   ;
\draw    (288.01,116) -- (315.25,116) ;
\draw [shift={(317.25,116)}, rotate = 180] [color={rgb, 255:red, 0; green, 0; blue, 0 }  ][line width=0.75]    (10.93,-3.29) .. controls (6.95,-1.4) and (3.31,-0.3) .. (0,0) .. controls (3.31,0.3) and (6.95,1.4) .. (10.93,3.29)   ;
\draw    (288.01,141) -- (315.25,141) ;
\draw [shift={(317.25,141)}, rotate = 180] [color={rgb, 255:red, 0; green, 0; blue, 0 }  ][line width=0.75]    (10.93,-3.29) .. controls (6.95,-1.4) and (3.31,-0.3) .. (0,0) .. controls (3.31,0.3) and (6.95,1.4) .. (10.93,3.29)   ;
\draw    (288.01,102) -- (316.19,89.23) ;
\draw [shift={(318.01,88.4)}, rotate = 155.61] [color={rgb, 255:red, 0; green, 0; blue, 0 }  ][line width=0.75]    (10.93,-3.29) .. controls (6.95,-1.4) and (3.31,-0.3) .. (0,0) .. controls (3.31,0.3) and (6.95,1.4) .. (10.93,3.29)   ;
\draw    (288.01,156.96) -- (316.19,173.89) ;
\draw [shift={(317.9,174.93)}, rotate = 211.01] [color={rgb, 255:red, 0; green, 0; blue, 0 }  ][line width=0.75]    (10.93,-3.29) .. controls (6.95,-1.4) and (3.31,-0.3) .. (0,0) .. controls (3.31,0.3) and (6.95,1.4) .. (10.93,3.29)   ;
\draw    (359.01,177) -- (386.25,177) ;
\draw [shift={(388.25,177)}, rotate = 180] [color={rgb, 255:red, 0; green, 0; blue, 0 }  ][line width=0.75]    (10.93,-3.29) .. controls (6.95,-1.4) and (3.31,-0.3) .. (0,0) .. controls (3.31,0.3) and (6.95,1.4) .. (10.93,3.29)   ;
\draw    (359.01,80) -- (386.25,80) ;
\draw [shift={(388.25,80)}, rotate = 180] [color={rgb, 255:red, 0; green, 0; blue, 0 }  ][line width=0.75]    (10.93,-3.29) .. controls (6.95,-1.4) and (3.31,-0.3) .. (0,0) .. controls (3.31,0.3) and (6.95,1.4) .. (10.93,3.29)   ;
\draw    (359.01,112) -- (387.19,99.23) ;
\draw [shift={(389.01,98.4)}, rotate = 155.61] [color={rgb, 255:red, 0; green, 0; blue, 0 }  ][line width=0.75]    (10.93,-3.29) .. controls (6.95,-1.4) and (3.31,-0.3) .. (0,0) .. controls (3.31,0.3) and (6.95,1.4) .. (10.93,3.29)   ;
\draw    (359.01,143.96) -- (387.19,160.89) ;
\draw [shift={(388.9,161.93)}, rotate = 211.01] [color={rgb, 255:red, 0; green, 0; blue, 0 }  ][line width=0.75]    (10.93,-3.29) .. controls (6.95,-1.4) and (3.31,-0.3) .. (0,0) .. controls (3.31,0.3) and (6.95,1.4) .. (10.93,3.29)   ;
\draw    (429.01,86) -- (457.19,102.93) ;
\draw [shift={(458.9,103.97)}, rotate = 211.01] [color={rgb, 255:red, 0; green, 0; blue, 0 }  ][line width=0.75]    (10.93,-3.29) .. controls (6.95,-1.4) and (3.31,-0.3) .. (0,0) .. controls (3.31,0.3) and (6.95,1.4) .. (10.93,3.29)   ;
\draw    (429.01,174) -- (457.19,161.23) ;
\draw [shift={(459.01,160.4)}, rotate = 155.61] [color={rgb, 255:red, 0; green, 0; blue, 0 }  ][line width=0.75]    (10.93,-3.29) .. controls (6.95,-1.4) and (3.31,-0.3) .. (0,0) .. controls (3.31,0.3) and (6.95,1.4) .. (10.93,3.29)   ;
\draw    (217,141) -- (226.39,141) -- (244.23,141) ;
\draw [shift={(246.23,141)}, rotate = 180] [color={rgb, 255:red, 0; green, 0; blue, 0 }  ][line width=0.75]    (10.93,-3.29) .. controls (6.95,-1.4) and (3.31,-0.3) .. (0,0) .. controls (3.31,0.3) and (6.95,1.4) .. (10.93,3.29)   ;

\draw (328.01,71.37) node [anchor=north west][inner sep=0.75pt]    {$\PIIIprn{H}$};
\draw (325.01,105.4) node [anchor=north west][inner sep=0.75pt]    {$\PPIVn{6}$};
\draw (325.01,133.4) node [anchor=north west][inner sep=0.75pt]    {$\PPIIIprn{1}$};
\draw (325.01,167.4) node [anchor=north west][inner sep=0.75pt]    {$\PPIVn{5}$};
\draw (254.01,119.4) node [anchor=north west][inner sep=0.75pt]    {$\PPVn{7}$};
\draw (181.01,105.4) node [anchor=north west][inner sep=0.75pt]    {\ref{eq:P6_14}};
\draw (401.01,71.39) node [anchor=north west][inner sep=0.75pt]    {$\PIIn{H}$};
\draw (395.01,167.4) node [anchor=north west][inner sep=0.75pt]    {$\PPIIn{1}$};
\draw (465.01,119.4) node [anchor=north west][inner sep=0.75pt]    {$\PIn{H}$};
\draw (181,133.4) node [anchor=north west][inner sep=0.75pt]    {\ref{eq:P6_17}};

\end{tikzpicture}}
    \end{figure}
    \end{minipage}
    \begin{minipage}[l]{0.49\linewidth}
    \begin{figure}[H]
        \hspace{0.3cm}
        \scalebox{0.93}{\tikzset{every picture/.style={line width=0.75pt}} 

\begin{tikzpicture}[x=0.75pt,y=0.75pt,yscale=-1,xscale=1]

\draw    (218.01,129) -- (227.4,129) -- (245.25,129) ;
\draw [shift={(247.25,129)}, rotate = 180] [color={rgb, 255:red, 0; green, 0; blue, 0 }  ][line width=0.75]    (10.93,-3.29) .. controls (6.95,-1.4) and (3.31,-0.3) .. (0,0) .. controls (3.31,0.3) and (6.95,1.4) .. (10.93,3.29)   ;
\draw    (288.01,116) -- (315.25,116) ;
\draw [shift={(317.25,116)}, rotate = 180] [color={rgb, 255:red, 0; green, 0; blue, 0 }  ][line width=0.75]    (10.93,-3.29) .. controls (6.95,-1.4) and (3.31,-0.3) .. (0,0) .. controls (3.31,0.3) and (6.95,1.4) .. (10.93,3.29)   ;
\draw    (288.01,141) -- (315.25,141) ;
\draw [shift={(317.25,141)}, rotate = 180] [color={rgb, 255:red, 0; green, 0; blue, 0 }  ][line width=0.75]    (10.93,-3.29) .. controls (6.95,-1.4) and (3.31,-0.3) .. (0,0) .. controls (3.31,0.3) and (6.95,1.4) .. (10.93,3.29)   ;
\draw    (288.01,102) -- (316.19,89.23) ;
\draw [shift={(318.01,88.4)}, rotate = 155.61] [color={rgb, 255:red, 0; green, 0; blue, 0 }  ][line width=0.75]    (10.93,-3.29) .. controls (6.95,-1.4) and (3.31,-0.3) .. (0,0) .. controls (3.31,0.3) and (6.95,1.4) .. (10.93,3.29)   ;
\draw    (288.01,156.96) -- (316.19,173.89) ;
\draw [shift={(317.9,174.93)}, rotate = 211.01] [color={rgb, 255:red, 0; green, 0; blue, 0 }  ][line width=0.75]    (10.93,-3.29) .. controls (6.95,-1.4) and (3.31,-0.3) .. (0,0) .. controls (3.31,0.3) and (6.95,1.4) .. (10.93,3.29)   ;
\draw    (359.01,177) -- (386.25,177) ;
\draw [shift={(388.25,177)}, rotate = 180] [color={rgb, 255:red, 0; green, 0; blue, 0 }  ][line width=0.75]    (10.93,-3.29) .. controls (6.95,-1.4) and (3.31,-0.3) .. (0,0) .. controls (3.31,0.3) and (6.95,1.4) .. (10.93,3.29)   ;
\draw    (359.01,80) -- (386.25,80) ;
\draw [shift={(388.25,80)}, rotate = 180] [color={rgb, 255:red, 0; green, 0; blue, 0 }  ][line width=0.75]    (10.93,-3.29) .. controls (6.95,-1.4) and (3.31,-0.3) .. (0,0) .. controls (3.31,0.3) and (6.95,1.4) .. (10.93,3.29)   ;
\draw    (359.01,112) -- (387.19,99.23) ;
\draw [shift={(389.01,98.4)}, rotate = 155.61] [color={rgb, 255:red, 0; green, 0; blue, 0 }  ][line width=0.75]    (10.93,-3.29) .. controls (6.95,-1.4) and (3.31,-0.3) .. (0,0) .. controls (3.31,0.3) and (6.95,1.4) .. (10.93,3.29)   ;
\draw    (359.01,143.96) -- (387.19,160.89) ;
\draw [shift={(388.9,161.93)}, rotate = 211.01] [color={rgb, 255:red, 0; green, 0; blue, 0 }  ][line width=0.75]    (10.93,-3.29) .. controls (6.95,-1.4) and (3.31,-0.3) .. (0,0) .. controls (3.31,0.3) and (6.95,1.4) .. (10.93,3.29)   ;
\draw    (429.01,86) -- (457.19,102.93) ;
\draw [shift={(458.9,103.97)}, rotate = 211.01] [color={rgb, 255:red, 0; green, 0; blue, 0 }  ][line width=0.75]    (10.93,-3.29) .. controls (6.95,-1.4) and (3.31,-0.3) .. (0,0) .. controls (3.31,0.3) and (6.95,1.4) .. (10.93,3.29)   ;
\draw    (429.01,174) -- (457.19,161.23) ;
\draw [shift={(459.01,160.4)}, rotate = 155.61] [color={rgb, 255:red, 0; green, 0; blue, 0 }  ][line width=0.75]    (10.93,-3.29) .. controls (6.95,-1.4) and (3.31,-0.3) .. (0,0) .. controls (3.31,0.3) and (6.95,1.4) .. (10.93,3.29)   ;

\draw (325.01,71.37) node [anchor=north west][inner sep=0.75pt]    {\ref{eq:P3_6}};
\draw (325.01,105.4) node [anchor=north west][inner sep=0.75pt]    {$\PPIVn{1}$};
\draw (325.01,133.4) node [anchor=north west][inner sep=0.75pt]    {$\PPIIIprn{8}$};
\draw (325.01,167.4) node [anchor=north west][inner sep=0.75pt]    {\ref{eq:P4_m2_m2}};
\draw (255.01,119.4) node [anchor=north west][inner sep=0.75pt]    {\ref{eq:P5_9}};
\draw (185.01,119.4) node [anchor=north west][inner sep=0.75pt]    {\ref{eq:P6_16}};
\draw (401.01,71.39) node [anchor=north west][inner sep=0.75pt]    {$\PIIn{H}$};
\draw (395.01,167.4) node [anchor=north west][inner sep=0.75pt]    {$\PPIIn{1}$};
\draw (465.01,119.4) node [anchor=north west][inner sep=0.75pt]    {$\PIn{H}$};

\end{tikzpicture}}
    \end{figure}
    \end{minipage}
    \caption{Degenerations of \ref{eq:P6_17}, \ref{eq:P6_14}, and \ref{eq:P6_16}}
    \label{pic:deg_orbit4}
\end{figure}
The bold font in $\PPkn{i}$ denotes the \PPainleve \, systems from Appendices  
 \ref{sec:sysintlistP6} -- \ref{sec:sysintlistP2}, the thin font is used for systems with Okamoto integrals, and the Hamiltonian systems are denoted as $\Pkn{H} $. 

In order to get degeneracies for the systems, obtained by the transposition from the systems from Appendix \ref{sec:sysintlistP6}, we need to replace each system in Figures \ref{pic:deg_orbit13} -- \ref{pic:deg_orbit4} with its transposed version. 

We see that the trees corresponding to $\PVI$ systems from different orbits (see Subsection \ref{sec:trgroup_P6}) have no intersections. All 18 systems listed in the appendices are presented in Figures \ref{pic:deg_orbit13} -- \ref{pic:deg_orbit4}. No additional new systems were obtained by degenerations, so the class of systems found in this paper is closed with respect to the limiting transitions. 

These limiting transitions can be also extended to the Lax pairs. In most cases, formulas from Subsection 4.1 in \cite{bobrova2022non} are applicable, but in the cases $\PV \to \PIIIpr$ and $\PIV \to \PII$ it is necessary to supplement them with a gauge transformation 
\begin{align}
    \tilde{\mathbf{A}} (z, \lambda)
    &= g \, \mathbf{A} (z, \lambda) \, g^{-1}
    + \partial_{\lambda} g \,\, g^{-1},
    &
    \tilde{\mathbf{B}} (z, \lambda)
    &= g \, \mathbf{B} (z, \lambda) \, g^{-1}
    + \partial_{z} g \,\, g^{-1},
\end{align}
where $g = z^{\mu} \, \mathbf{I}$ and $g = z^{\mu \, \varepsilon^{-6}} \, \mathbf{I}$, respectively. The constant $\mu$ depends on the choice of a particular system of $\PV$ or  $\PIV$ type. 

\section{Conclusion}

We found all non-abelian \Painleve systems whose auxiliary autonomous systems have first integrals and symmetries. We believe that these systems enjoy the \Painleve property. It was shown in \cite{Balandin_Sokolov_1998, Adler_Sokolov_2020_1, Bobrova_Sokolov_2022} that in the $\PI$, $\PII$, and $\PIV$ cases they satisfy the \PPainleve-Kovalevskaya~test. 

We constructed isomonodromic Lax pairs for multiparametric families of $\PII$ -- $\PVI$ systems containing all the systems obtained in this paper. As long as we assume that the coefficients of non-abelian polynomials are scalar constants, all systems are  $GL(m)$ invariant and we may consider their restriction on the invariants of $GL(m)$-action. It turns out that all restricted systems of the above families coincide with each other. 

A comparison of the results from \cite{Bobrova_Sokolov_2022} and Subsection \ref{P4} shows that in order to choose from these families those systems that have the \Painleve property, one can assume the existence of symmetries for the auxiliary autonomous systems.

The situation changes drastically if the constant coefficients themselves are non-abelian. Examples of such \Painleve systems were found in \cite{Balandin_Sokolov_1998, Retakh_Rubtsov_2010, Adler_Kolesnikov_2022, Bobrova_Sokolov_2022, bobrova2022classification}. We hope that a systematic construction of such systems can be carried out by simultaneous non-abelianization of the coefficients of the autonomous systems found in this paper and their symmetries.

We are going to devote a separate paper to the search for systems of the \Painleve type with non-abelian coefficients and the study of the properties of the corresponding autonomous systems.  
 
\subsubsection*{Acknowledgements}

The authors are grateful to V.~E.~Adler for useful discussions. They are thankful to IHES for hospitality. The research of the second author was carried out under the State Assignment 0029-2021-0004 (Quantum field theory) of the Ministry of Science and Higher Education of the Russian Federation. The first author was partially supported by the International Laboratory of Cluster Geometry HSE, RF Government grant № 075-15-2021-608, and by Young Russian Mathematics award.
 
\appendix

\section*{Appendices}
\section{Lists of non-abelian autonomous systems of \Painleve type}
\label{sec:sysintlist}
  
\subsection{Systems of \texorpdfstring{$\PVI$}{P6} type}\label{sec:sysintlistP6}


\begin{align}
    \label{eq:P6_1}
    \tag*{$\PPPVIn{1}$}
    &\begin{gathered}
    \left\{
    \begin{array}{lcr}
         u'
         &=& - u^3 v + 3 u^2 v u
         - 2 u v u
         - \kappa_1 u^2 
         + \kappa_2 u
         \hspace{4.6cm}
         \\[1mm]
         && 
         + \, z \brackets{
         - 2 u v u
         + u v + v u
         + \kappa_4 u 
         + (\kappa_1 - \kappa_2 - \kappa_4)
         },
         \\[2mm]
         v'
         &=& u^2 v^2 - 2 u v u v - 2 v u^2 v
         + 2 v u v
         + (\kappa_1 - \hat \kappa_3) u v 
         + (\kappa_1 + \hat \kappa_3) v u
         - \kappa_2 v
         \\[1mm]
         && 
         + \, \tfrac14 (
         \hat \kappa_3^2 
         - \kappa_1^2
         )
         + z \brackets{
         2 v u v
         - v^2
         - \kappa_4 v
         },
    \end{array}
    \right.
    \\[2mm]
    \kappa_3
    = \tfrac14 (
    \hat \kappa_3^2 
    - \kappa_1^2
    ).
    \end{gathered}
    \\[4mm]
    \label{eq:P6_2}
    \tag*{$\PPPVIn{2}$}
    &\left\{
    \begin{array}{lcr}
         u'
         &=& u^3 v + u^2 v u
         - 2 u^2 v
         - \kappa_1 u^2 
         + \kappa_2 u 
         \hspace{5.2cm}
         \\[1mm]
         && 
         + \, z \brackets{
         - 2 u^2 v
         + 3 u v
         - v u
         + \kappa_4 u 
         + (\kappa_1 - \kappa_2 - \kappa_4)
         },
         \\[2mm]
         v'
         &=& - u^2 v^2 - 2 u v u v
         + 2 u v^2
         + 2 \kappa_1 u v
         - \kappa_2 v
         + \kappa_3
         + z \brackets{
         2 u v^2
         - v^2
         - \kappa_4 v
         }.
    \end{array}
    \right.
    \\[4mm]
    \label{eq:P6_4}
    \tag*{$\PPPVIn{3}$}
    &\left\{
    \begin{array}{lcr}
         u'
         &=& u^3 v + u^2 v u
         - 2 u^2 v
         - \kappa_1 u^2 
         + \kappa_2 u 
         \hspace{5.2cm}
         \\[1mm]
         && 
         + \, z \brackets{
         - 2 u v u
         + u v
         + v u
         + \kappa_4 u 
         + (\kappa_1 - \kappa_2 - \kappa_4)
         },
         \\[2mm]
         v'
         &=& - u^2 v^2 - 2 u v u v
         + 2 u v^2
         + 2 \kappa_1 u v
         - \kappa_2 v
         + \kappa_3
         + z \brackets{
         2 v u v
         - v^2
         - \kappa_4 v
         }.
    \end{array}
    \right.
    \\[4mm]
    \label{eq:P6_7}
    \tag*{$\PPPVIn{4}$}
    &\left\{
    \begin{array}{lcr}
         u'
         &=& u^3 v + u^2 v u
         - 2 u v u
         - \kappa_1 u^2 
         + \kappa_2 u 
         \hspace{5.2cm}
         \\[1mm]
         && 
         + \, z \brackets{
         - 2 u^2 v
         + u v
         + v u
         + \kappa_4 u 
         + (\kappa_1 - \kappa_2 - \kappa_4)
         },
         \\[2mm]
         v'
         &=& - u^2 v^2 - 2 u v u v
         + 2 v u v
         + 2 \kappa_1 u v
         - \kappa_2 v
         + \kappa_3
         + z \brackets{
         2 u v^2
         - v^2
         - \kappa_4 v
         }.
    \end{array}
    \right.
    \\[4mm]
    \label{eq:P6_8}
    \tag*{$\PPPVIn{5}$}
    &\begin{gathered}
    \left\{
    \begin{array}{lcr}
         u'
         &=&
         2 u^3 v
         - 2 u^2 v
         - \kappa_1 u^2 
         + \kappa_2 u 
         + z \brackets{
         - 2 u^2 v
         + 2 u v
         + \kappa_4 u 
         + (\kappa_1 - \kappa_2 - \kappa_4)
         },
         \\[2mm]
         v'
         &=&
         - 2 u^2 v^2 - 2 u v u v + v u^2 v
         + 2 u v^2
         + \tfrac12 (5 \kappa_1 - \hat \kappa_3) u v
         + \tfrac12 (- \kappa_1 + \hat \kappa_3) v u
         \hspace{6mm}
         \\[1mm]
         && 
         - \, \kappa_2 v
         + \tfrac14 (\hat \kappa_3^2 - \kappa_1^2)
         + z \brackets{
         2 u v^2
         - v^2
         - \kappa_4 v
         },
    \end{array}
    \right.
    \\[2mm]
    \kappa_3
    = \tfrac14 (
    \hat \kappa_3^2 
    - \kappa_1^2
    ).
    \end{gathered}
    \\[4mm]
    \label{eq:P6_17}
    \tag*{$\PPPVIn{6}$}
    &\begin{gathered}
    \left\{
    \begin{array}{lcr}
         u'
         &=& 2 u^2 v u
         - 2 u^2 v
         - \kappa_1 u^2 
         + \kappa_2 u 
         + z \brackets{
         - 2 u v u
         + 2 u v
         + \kappa_4 u 
         + (\kappa_1 - \kappa_2 - \kappa_4)
         },
         \\[2mm]
         v'
         &=& - 2 u v u v - v u^2 v
         + 2 u v^2
         + \tfrac12 (3 \kappa_1 - \hat \kappa_3) u v
         + \tfrac12 (\kappa_1 + \hat \kappa_3) v u
         - \kappa_2 v
         \hspace{1.5cm}
         \\[1mm]
         && 
         + \, \tfrac14 (\hat \kappa_3^2 - \kappa_1^2)
         + z \brackets{
         2 v u v
         - v^2
         - \kappa_4 v
         },
    \end{array}
    \right.
    \\[2mm]
    \kappa_3
    = \tfrac14 (
    \hat \kappa_3^2 
    - \kappa_1^2
    ).
    \end{gathered}
    \\[4mm]
    \label{eq:P6_14}
    \tag*{$\PPPVIn{7}$}
    &\begin{gathered}
    \left\{
    \begin{array}{lcr}
         u'
         &=& 2 u^2 v u
         - 2 u v u
         - \kappa_1 u^2 
         + \kappa_2 u 
         + z \brackets{
         - 2 u^2 v
         + 2 u v
         + \kappa_4 u 
         + (\kappa_1 - \kappa_2 - \kappa_4)
         },
         \\[2mm]
         v'
         &=& - 2 u v u v - v u^2 v
         + 2 v u v
         + \tfrac12 (3 \kappa_1 - \hat \kappa_3) u v
         + \tfrac12 (\kappa_1 + \hat \kappa_3) v u
         - \kappa_2 v
         \hspace{1.5cm}
         \\[1mm]
         && 
         + \, \tfrac14 (\hat \kappa_3^2 - \kappa_1^2)
         + z \brackets{
         2 u v^2
         - v^2
         - \kappa_4 v
         },
    \end{array}
    \right.
    \\[2mm]
    \kappa_3
    = \tfrac14 (
    \hat \kappa_3^2 
    - \kappa_1^2
    ).
    \end{gathered}
    \\[4mm]
    \label{eq:P6_16}
    \tag*{$\PPPVIn{8}$}
    &\begin{gathered}
    \left\{
    \begin{array}{lcr}
         u'
         &=& 2 u^2 v u
         - 2 u v u
         - \kappa_1 u^2 
         + \kappa_2 u 
         + z \brackets{
         - 2 u v u
         + 2 v u
         + \kappa_4 u 
         + (\kappa_1 - \kappa_2 - \kappa_4)
         },
         \\[2mm]
         v'
         &=& - 2 u v u v - v u^2 v
         + 2 v u v
         + \tfrac12 (3 \kappa_1 - \hat \kappa_3) u v
         + \tfrac12 (\kappa_1 + \hat \kappa_3) v u
         - \kappa_2 v
         \hspace{1.5cm}
         \\[1mm]
         && 
         + \, \tfrac14 (\hat \kappa_3^2 - \kappa_1^2)
         + z \brackets{
         2 v u v
         - v^2
         - \kappa_4 v
         },
    \end{array}
    \right.
    \\[2mm]
    \kappa_3
    = \tfrac14 (
    \hat \kappa_3^2 
    - \kappa_1^2
    ).
    \end{gathered}
\end{align}

\subsection{Systems of \texorpdfstring{$\PV$}{P5} type}
\label{sec:sysintlistP5}

\begin{align}
    \label{eq:P5_1}
    \tag*{$\PPPVn{1}$}
    &\begin{gathered}
    \left\{
    \begin{array}{lcr}
         u'
         &=& - u^3 v + 3 u^2 v u
         - 4 u v u
         - \kappa_1 u^2 
         + u v + v u 
         + (\kappa_1 + \kappa_2) u
         - \kappa_2 
         + \kappa_4 z u,
         \\[2mm]
         v'
         &=& u^2 v^2 - 2 u v u v - 2 v u^2 v
         + 4 v u v
         + (\kappa_1 - \hat \kappa_3) u v 
         + (\kappa_1 + \hat \kappa_3) v u
         - v^2
         \hspace{4mm}
         \\[1mm]
         &&
         - \, (\kappa_1 + \kappa_2) v
         + \tfrac14 (\hat \kappa_3^2 - \kappa_1^2)
         - \kappa_4 z v,
    \end{array}
    \right.
    \\[2mm]
    \kappa_3
    = \tfrac14 (
    \hat \kappa_3^2 
    - \kappa_1^2
    ).
    \end{gathered}
    \\[4mm]
    \label{eq:P5_2}
    \tag*{$\PPPVn{2}$}
    &\left\{
    \begin{array}{lcl}
         u'
         &=& u^3 v + u^2 v u
         - 4 u^2 v
         - \kappa_1 u^2 
          + 3 u v - v u
         + (\kappa_1 + \kappa_2) u
         - \kappa_2 
         + \kappa_4 z u,
         \\[2mm]
         v'
         &=& - u^2 v^2 - 2 u v u v
         + 4 u v^2
         + 2 \kappa_1 u v
         - v^2
         - (\kappa_1 + \kappa_2) v
         + \kappa_3
         - \kappa_4 z v.
    \end{array}
    \right.
    \\[4mm]
    \label{eq:P5_5}
    \tag*{$\PPPVn{3}$}
    &\begin{gathered}
    \left\{
    \begin{array}{lcr}
         u'
         &=& 2 u^3 v
         - 4 u^2 v
         - \kappa_1 u^2 
         + 2 u v
         + (\kappa_1 + \kappa_2) u
         - \kappa_2 
         + \kappa_4 z u,
         \hspace{20mm}
         \\[2mm]
         v'
         &=& - 2 u^2 v^2 - 2 u v u v + v u^2 v
         + 4 u v^2
         + \tfrac12 (5 \kappa_1 - \hat \kappa_3) u v 
         + \tfrac12 (- \kappa_1 + \hat \kappa_3) v u
         \\[1mm]
         && 
         - \, v^2
         - (\kappa_1 + \kappa_2) v
         + \tfrac14 (\hat \kappa_3^2 - \kappa_1^2)
         - \kappa_4 z v,
    \end{array}
    \right.
    \\[2mm]
    \kappa_3
    = \tfrac14 (
    \hat \kappa_3^2 
    - \kappa_1^2
    ).
    \end{gathered}
    \\[4mm]
    \label{eq:P5_9}
    \tag*{$\PPPVn{4}$}
    &\begin{gathered}
    \left\{
    \begin{array}{lcr}
         u'
         &=& 2 u^2 v u
         - 4 u v u
         - \kappa_1 u^2 
         + 2 v u
         + (\kappa_1 + \kappa_2) u
         - \kappa_2 
         + \kappa_4 z u,
         \hspace{2mm}
         \\[2mm]
         v'
         &=& - 2 u v u v - v u^2 v
         + 4 v u v
         + \tfrac12 (3 \kappa_1 - \hat \kappa_3) u v 
         + \tfrac12 (\kappa_1 + \hat \kappa_3) v u
         \\[1mm]
         && 
         - \, v^2
         - (\kappa_1 + \kappa_2) v
         + \tfrac14 (\hat \kappa_3^2 - \kappa_1^2)
         - \kappa_4 z v.
    \end{array}
    \right.
    \\[2mm]
    \kappa_3
    = \tfrac14 (
    \hat \kappa_3^2 
    - \kappa_1^2
    ).
    \end{gathered}
\end{align}

\subsection{Systems of \texorpdfstring{$\PIV$}{P4} type}
\label{sec:sysintlistP4}
\begin{align}
    \label{eq:P4_m2_m2}
    \tag*{$\PPPIVn{1}$}
    &\left\{
    \begin{array}{lcl}
         u' 
         &=& 
         - u^2 
         + 2 v u
         - 2 z u
         + \kappa_2,
         \\[2mm]
         v'
         &=& 
         - v^2
         + 2 u v
         + 2 z v
         + \kappa_3.
    \end{array}
    \right.
    \\[4mm]
    \label{eq:P4_0_m3}
    \tag*{$\PPPIVn{2}$}
    &\left\{
    \begin{array}{lcl}
         u' 
         &=& 
         - u^2 
         + 2 u v
         - 2 z u
         + \kappa_2,
         \\[2mm]
         v'
         &=& 
         - v^2
         + 3 u v
         - v u
         + 2 z v
         + \kappa_3.
    \end{array}
    \right.
    \\[4mm]
    \label{eq:P4_1_m2}
    \tag*{$\PPPIVn{3}$}
    &\left\{
    \begin{array}{lcl}
         u' 
         &=& 
         - u^2 
         + 3 u v - v u
         - 2 z u
         + \kappa_2,
         \\[2mm]
         v'
         &=& 
         - v^2
         + 2 u v 
         + 2 z v
         + \kappa_3.
    \end{array}
    \right.
\end{align}

\subsection{Systems of \texorpdfstring{$\PIIIpr$}{P3'} type}
\label{sec:sysintlistP3}

\begin{align}
    \label{eq:P3_4}
    \tag*{$\PPPIIIprn{1}$}
    &\left\{
    \begin{array}{lcl}
         u'
         &=& 2 u^2 v 
         + \kappa_1 u
         + \kappa_2 u^2 
         + \kappa_4 z,
         \\[2mm]
         v'
         &=& - 2 u v^2 
         - \kappa_1 v
         - 3 \kappa_2 u v 
         + \kappa_2 v u 
         - \kappa_3.
    \end{array}
    \right.
    \\[4mm]
    \label{eq:P3_6}
    \tag*{$\PPPIIIprn{2}$}
    &\left\{
    \begin{array}{lcl}
         u'
         &=& 2 u^2 v 
         + \kappa_1 u
         + \kappa_2 u^2 
         + \kappa_4 z,
         \\[2mm]
         v'
         &=& - 2 u v^2 
         - \kappa_1 v
         - \kappa_2 u v 
         - \kappa_2 v u 
         - \kappa_3.
    \end{array}
    \right.
\end{align}

\subsection{Systems of \texorpdfstring{$\PII$}{P2} type}
\label{sec:sysintlistP2}
\begin{gather} 
    \label{eq:P2_m3}
    \tag*{$\PPPIIn{1}$}
    \left\{
    \begin{array}{lcl}
         u'
         &=& - u^2 
         + v
         - \tfrac12 z,
         \\[2mm]
         v'
         &=& 3 u v - v u
         + \kappa_3.
    \end{array}
    \right.
\end{gather}

\section{Special case of \texorpdfstring{$\PIIIpr$}{P3'} type systems}\label{P3'D7}

\subsection{Non-abelian systems of \texorpdfstring{$\PIIIpr(D_7)$}{P3'(D7)} type}
\label{sec:sysintP3D72}

For the same reasons as in the papers \cite{bobrova2022non} and \cite{bobrova2022classification}, here we consider a special case $\PIIIpr(D_7)$ of the \PPainleve-$3^{\prime}$ system.   The corresponding scalar autonomous system is Hamiltonian with
\begin{equation}
    \label{eq:ham3d7}
    H
    = u^2 v^2 
    + \kappa_2 u^2 v
    + \kappa_1 u v 
    + \kappa_3 u 
    + \kappa_4 z u^{-1}.
\end{equation}
A non-abelian generalization of the autonomous system reads as
\begin{align}
    \label{eq:sysP3D7_nc}
    \left\{
    \begin{array}{lcr}
         \dfrac{du}{dt}
         &=& a_1 u^2 v 
         + (2 - a_1 - a_2) u v u
         + a_2 v u^2
         + \kappa_1 u + \kappa_2 u^2
         ,
         \\[3mm]
         \dfrac{dv}{dt}
         &=& b_1 u v^2 
         - (2 + b_1 + b_2) v u v
         + b_2 v^2 u
         - \kappa_1 v
         \hspace{1.5cm}
         \\
         &&\quad
         + \, c_1 u v
         + (- 2 \kappa_2 - c_1) v u
         - \kappa_3 + \kappa_4 z u^{-2},
    \end{array}
    \right.
\end{align}
where $a_i$, $b_i$, $c_1$, $\kappa_i \in \mathbb{C}$.

\begin{proposition}
\label{thm:intP3D7_nc}
\phantom{}

\begin{itemize}
    \item[\rm{a)}]
    A system of the form \eqref{eq:sysP3D7_nc} admits a non-abelianization of integrals $H$ and $H^2$ iff
    \begin{align}
        &&
        b_1
        &= - a_1,
        &&&
        b_2
        &= - a_2.
        &&
    \end{align}

    \item[\rm{b)}]
    If the system additionally possesses a non-abelianization of the symmetry \eqref{scsymgen} of degree 7, then it belongs to one of the following lists:
\begin{enumerate}
    \item[\rm{i)}] 
    Appendix A.7 from {\rm\cite{bobrova2022classification}},
    
    \item[\rm{ii)}] 
    Appendix C.1 from {\rm\cite{bobrova2022non}},
    
    \item[\rm{iii)}] 
    {\rm(}up to the transposition $\tau${\rm)}
    \begin{align}
        \label{eq:P3D7_4}
        \tag*{$\PPPIIIprn{1}(D_7)$}
        &\left\{
        \begin{array}{lcl}
             u'
             &=& 2 u^2 v 
             + \kappa_1 u
             + \kappa_2 u^2,
             \\[2mm]
             v'
             &=& - 2 u v^2 
             - \kappa_1 v
             - 3 \kappa_2 u v 
             + \kappa_2 v u 
             - \kappa_3
             + \kappa_4 z u^{-2}
             ;
        \end{array}
        \right.
        \\[4mm]
        \label{eq:P3D7_6}
        \tag*{$\PPPIIIprn{2}(D_7)$}
        &\left\{
        \begin{array}{lcl}
             u'
             &=& 2 u^2 v 
             + \kappa_1 u
             + \kappa_2 u^2,
             \\[2mm]
             v'
             &=& - 2 u v^2 
             - \kappa_1 v
             - \kappa_2 u v 
             - \kappa_2 v u 
             - \kappa_3
             + \kappa_4 z u^{-2}
             .
        \end{array}
        \right.
    \end{align}
\end{enumerate}
\end{itemize}
\end{proposition}

The transposition $\tau$ acts on the parameters of system \eqref{eq:sysP3D7_nc} as follows
\begin{align}
    \tau \brackets{
    a_1, \,
    a_2, \,
    c_1
    }
    &= \brackets{
    a_2, \,
    a_1, \,
    - c_1 - 2 \kappa_2
    }.
\end{align}
Systems from the list iii) are non-equivalent with respect to $\tau$. There non-equivalent systems with Okamoto integral were obtained in \cite{bobrova2022non}. One Hamiltonian system can be found in \cite{bobrova2022classification}.

The $\PIIIpr(D_7)$ system related to Hamiltonian \eqref{eq:ham3d7} has an isomonodromic Lax pair \cite{jimbo1981monodromy} of the form \eqref{eq:matABform_P3D6'_}.
Matrices $A_0 (z)$, $A_{-1} (z)$, $A_{-2} (z)$, $B_1$, and $B_0 (z)$ are equivalent to
\begin{gather}
    \notag
    \begin{aligned}
    A_0
    &= 
    \begin{pmatrix}
    0 & 0 \\[0.9mm]
    \kappa_4 \, z u^{-1} & 0
    \end{pmatrix},
    &
    A_{-1}
    &= -
    \begin{pmatrix}
    u v + \kappa_1
    & 
    u
    \\[0.9mm]
    \kappa_2 u v 
    + \kappa_3
    & 
    - u v
    \end{pmatrix},
    &
    A_{-2}
    &= 
    \begin{pmatrix}
    - \kappa_2
    & 
    1
    \\[0.9mm]
    0
    &
    0
    \end{pmatrix},
    \end{aligned}
    \\[-1mm]
    \label{eq:Laxpair_P3D7}
    \\[-1mm]
    \notag
    \begin{aligned}
    B_1
    &= 
    \begin{pmatrix}
    0 & 0 
    \\[0.9mm]
    \kappa_4 \, z & 0
    \end{pmatrix}
    ,
    &&&
    B_0
    &= z^{-1}
    \begin{pmatrix}
    \kappa_2 u 
    &
    - u
    \\[0.9mm]
    0
    & 
    0
    \end{pmatrix}
    .
    \end{aligned}
\end{gather}

One can generalize this pair to the non-commutative case for non-abelian \PPainleve-$3'(D_7)$ systems satisfying the item 3 from Assumption \ref{assumpt1} as follows.

\begin{proposition}
\label{thm:ncLax_genP3D7}
Suppose that  a non-abelian system of the $\PIIIpr(D_7)$ type admits an auxiliary autonomous system defined by the item {\rm{a)}} from Proposition {\rm\ref{thm:intP3D7_nc}}. Then it is equivalent to \eqref{eq:zerocurvcond}, where the Lax pair has the form \eqref{eq:matABform_P3D6'_} with
    \begin{gather} 
    \begin{aligned}
    A_0
    &= 
    \begin{pmatrix}
    0 & 0 \\[0.9mm]
    \kappa_4 \, z u^{-1} & 0
    \end{pmatrix},
    &
    A_{-1}
    &= -
    \begin{pmatrix}
    u v + \kappa_1
    & 
    u
    \\[0.9mm]
    \kappa_2 v u 
    + \kappa_3
    & 
    - v u
    \end{pmatrix},
    &
    A_{-2}
    &= 
    \begin{pmatrix}
    - \kappa_2
    & 
    1
    \\[0.9mm]
    0
    &
    0
    \end{pmatrix},
    \end{aligned}
    \\[2mm]
    \notag
    \begin{aligned}
    B_1
    &= 
    \begin{pmatrix}
    0 & 0 
    \\[0.9mm]
    \kappa_4 \, u^{-1} & 0
    \end{pmatrix},
    &&&
    B_{0}
    &= z^{-1}
    \begin{pmatrix}
    \kappa_2 u 
    &
    - u
    \\[0.9mm]
    0
    & 
    0
    \end{pmatrix}
    + q(z, u, v) \, \mathbf{I}
    ,
    \end{aligned}
    \\[2mm]
    \begin{aligned}
    z \, q (z, u, v)
    = - a_1 u v 
    + a_2 v u
    + c_1 u
    .
    \end{aligned}
\end{gather}
\vspace{-3mm}
\qed
\end{proposition}

\subsection{Limiting transitions}
\label{sec:deg_P3}
The formulas from Appendix C.2 in \cite{bobrova2022non} lead to the following degeneracy scheme of both systems and pairs:
\begin{figure}[H]
    \centering
    \scalebox{0.93}{\tikzset{every picture/.style={line width=0.75pt}} 

\begin{tikzpicture}[x=0.75pt,y=0.75pt,yscale=-1,xscale=1]

\draw [color={rgb, 255:red, 0; green, 0; blue, 0 }  ,draw opacity=1 ]   (53,100) -- (80.23,100) ;
\draw [shift={(82.23,100)}, rotate = 180] [color={rgb, 255:red, 0; green, 0; blue, 0 }  ,draw opacity=1 ][line width=0.75]    (10.93,-3.29) .. controls (6.95,-1.4) and (3.31,-0.3) .. (0,0) .. controls (3.31,0.3) and (6.95,1.4) .. (10.93,3.29)   ;
\draw [color={rgb, 255:red, 0; green, 0; blue, 0 }  ,draw opacity=1 ]   (53,131) -- (80.23,131) ;
\draw [shift={(82.23,131)}, rotate = 180] [color={rgb, 255:red, 0; green, 0; blue, 0 }  ,draw opacity=1 ][line width=0.75]    (10.93,-3.29) .. controls (6.95,-1.4) and (3.31,-0.3) .. (0,0) .. controls (3.31,0.3) and (6.95,1.4) .. (10.93,3.29)   ;
\draw [color={rgb, 255:red, 0; green, 0; blue, 0 }  ,draw opacity=1 ]   (162,100) -- (189.23,100) ;
\draw [shift={(191.23,100)}, rotate = 180] [color={rgb, 255:red, 0; green, 0; blue, 0 }  ,draw opacity=1 ][line width=0.75]    (10.93,-3.29) .. controls (6.95,-1.4) and (3.31,-0.3) .. (0,0) .. controls (3.31,0.3) and (6.95,1.4) .. (10.93,3.29)   ;
\draw [color={rgb, 255:red, 0; green, 0; blue, 0 }  ,draw opacity=1 ]   (162,131) -- (189.23,131) ;
\draw [shift={(191.23,131)}, rotate = 180] [color={rgb, 255:red, 0; green, 0; blue, 0 }  ,draw opacity=1 ][line width=0.75]    (10.93,-3.29) .. controls (6.95,-1.4) and (3.31,-0.3) .. (0,0) .. controls (3.31,0.3) and (6.95,1.4) .. (10.93,3.29)   ;

\draw (95,90.4) node [anchor=north west][inner sep=0.75pt]    {\ref{eq:P3D7_4}};
\draw (16,90.4) node [anchor=north west][inner sep=0.75pt]    {\ref{eq:P3_4}};
\draw (95,121.4) node [anchor=north west][inner sep=0.75pt]    {\ref{eq:P3D7_6}};
\draw (16,121.4) node [anchor=north west][inner sep=0.75pt]    {\ref{eq:P3_6}};
\draw (202,106.4) node [anchor=north west][inner sep=0.75pt]    {$\PIn{H}$};

\end{tikzpicture}}
    \caption{Degenerations of \ref{eq:P3_4} and \ref{eq:P3_6}}
    \label{pic:deg_P3}
\end{figure}

\bibliographystyle{plain}
\bibliography{bib}

\end{document}